\documentclass[epj]{svjour}

\usepackage{amsmath,amssymb,multirow}
\usepackage{graphicx}
\usepackage[usenames,dvipsnames]{color}
\usepackage[noend]{algorithmic}
\usepackage{color}
\usepackage{algorithm}
\usepackage[caption=false]{subfig}
\usepackage{url}
\usepackage{citesort}

\newcommand{\bbR}{\mathbb{R}}
\newcommand{\biu}{\boldsymbol{u}}
\newcommand{\biv}{\boldsymbol{v}}
\newcommand{\biw}{\boldsymbol{w}}

\newcommand{\eat}{\tau_{\mathrm{eat}}}
\newcommand{\ldt}{\tau_{\mathrm{ldt}}}

\newcommand{\TC}{C}

\newcommand{\E}{\mathop{\mathbf{E}}}

\begin{document}

\title{Coverage centralities for temporal networks}
\author{Taro Takaguchi\inst{1,}\inst{2,}\thanks{All the authors contributed equally to the work.}\thanks{\emph{\email{t\_takaguchi@nii.ac.jp}}} \and Yosuke Yano\inst{2,}\inst{3,}$^{\rm a}$ \and Yuichi Yoshida\inst{1,}\inst{4,}$^{\rm a}$
}
\institute{
National Institute of Informatics, 2-1-2 Hitotsubashi, Chiyoda-ku, Tokyo 101-8430, Japan \and
JST, ERATO, Kawarabayashi Large Graph Project, 2-1-2 Hitotsubashi, Chiyoda-ku, Tokyo 101-8430, Japan \and
Department of Computer Science, The University of Tokyo, 3-7-1 Hongo, Bunkyo-ku, Tokyo 113-8654, Japan \and
Preferred Infrastructure, Inc., 2-40-1 Hongo, Bunkyo-ku, Tokyo, 113-0033, Japan
}

\abstract{
Structure of real networked systems, such as social relationship, can be modeled as temporal networks in which each edge appears only at the prescribed time. Understanding the structure of temporal networks requires quantifying the importance of a temporal vertex, which is a pair of vertex index and time. In this paper, we define two centrality measures of a temporal vertex based on the fastest temporal paths which use the temporal vertex. The definition is free from parameters and robust against the change in time scale on which we focus. In addition, we can efficiently compute these centrality values for all temporal vertices. Using the two centrality measures, we reveal that distributions of these centrality values of real-world temporal networks are heterogeneous. For various datasets, we also demonstrate that a majority of the highly central temporal vertices are located within a narrow time window around a particular time. In other words, there is a bottleneck time at which most information sent in the temporal network passes through a small number of temporal vertices, which suggests an important role of these temporal vertices in spreading phenomena.
}
\maketitle

\section{Introduction}\label{sec:intro}
Complex networks such as social networks, information networks, and biological networks have been intensively studied in the past decade to understand their behavior under certain dynamics and develop efficient algorithms for them.
See~\cite{Albert:2002,Newman:2003da,Boccaletti:2006tu,Newman:2010} for extensive surveys.

However, many real-world networks are actually temporal networks~\cite{Holme:2012,Holme2013}, in which a vertex communicates with another vertex at specific time over finite duration.
For example, social interaction between individuals, passenger flow between cities, and synaptic transmission between neurons can be represented as temporal networks.
When we assume that the focal dynamical processes on networks, such as information propagation, occur on a time scale comparable to the change in network structure, a temporal-network representation gives us a precise way to capture the processes. 
We can describe the advantage of working with a temporal network using the example shown in Fig.~\ref{fig:temporal-network}.
This temporal network consists of four vertices and eight edges, each of which has the time it appears.
Let us assume that it takes unit time to send the information from the tail to the head of an edge.
For example, suppose that the information starts to propagate from $v_1$ at time $1$.
Then, it reaches $v_2$ at time $2$ through edge $(v_1,v_2)$, waits at $v_2$ till time $3$, then reaches $v_3$ at time $4$ through edge $(v_2,v_3)$.
The information never reaches $v_4$ because the only edge incoming to $v_4$ is $(v_2, v_4)$ which appears at time $1$,
and $v_2$ does not have the information at that time.
However, if we ignore the temporal information and regard the network as a static directed network,
we mistakenly reach the conclusion that information in $v_1$ at time $1$ can reach $v_4$ because there is a directed path from $v_1$ to $v_4$.
Therefore, we cannot dismiss temporal information to properly understand the structure of temporal networks.

\begin{figure}[!t]
  \centering
  \includegraphics[width = 0.48\hsize]{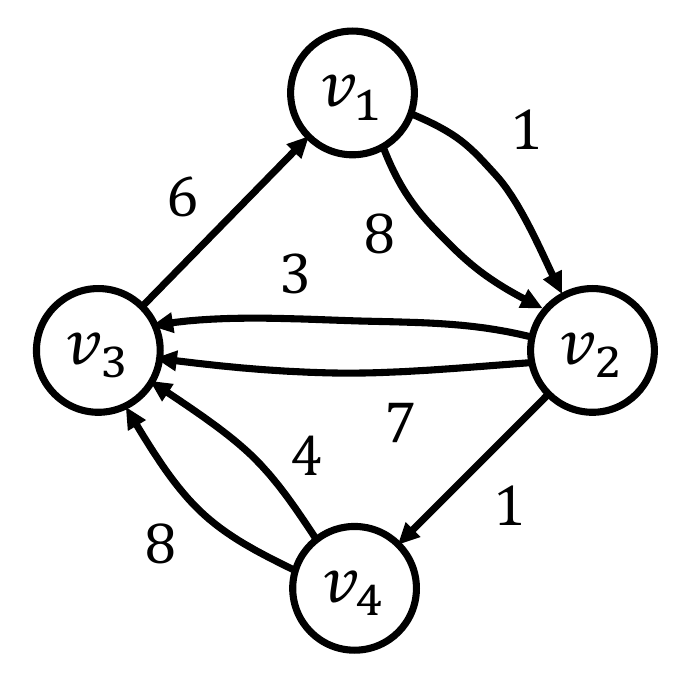}
  \caption{Schematic of an example of temporal network. The number associated with each edge represents the time at which the edge appears.}
  \label{fig:temporal-network}
\end{figure}

An important notion studied to understand the structure of (static) networks is vertex centrality, which measures the importance of a vertex.
The following reasons motivate the study of centralities.
First, we can use centralities to find important vertices in several applications such as suppressing the epidemics~\cite{Barrat:2008,Pastor-Satorras:2014} or maximizing the spread of influence~\cite{Kempe:2003iu}.
Second, we can use them to understand the structure of real-world networks by examining the difference between the distributions of the centrality values in such networks and in the randomized networks (e.g, \cite{Barrat:2004,Guimera:2005}).
Third, we can examine the validity of generative network models by investigating the distribution of centralities of the generated network (e.g., \cite{Kitsak:2007,Kumpula:2007}).

Hence, it is natural to study centralities for temporal networks.
Since the most fundamental difference between a static network and a temporal network is that the latter involves time,
we define the centrality of a vertex at a specific time.
To distinguish from a vertex,
we call the pair of a vertex and time a temporal vertex.
In the literature, multiple centrality notions of temporal vertices based on temporal paths~\cite{Holme:2012} have been proposed.
Examples include the generalizations of the centrality notions to temporal networks, such as betweenness~\cite{Tang:2010,Tang:2010go,Kim:2012,Alsayed:2015}, closeness~\cite{Tang:2010go,Pan:2011,Kim:2012}, communicability~\cite{Grindrod:2011,Estrada:2013,Grindrod:2013}, efficiency~\cite{Tang:2009}, random-walk centrality~\cite{Rocha:2014}, and win--lose score~\cite{Motegi:2012} (see Ref.~\cite{Nicosia:2013} for a review of some of them).
However, each previous centrality notion suffers from at least one of the following two issues:
\begin{enumerate}
\item We need to carefully set parameter values and (or) the time interval within which we consider temporal paths.
\item It is inefficient to compute the centrality.
\end{enumerate}

For the first issue, the time interval length especially requires careful tuning; if the time interval is too wide, then the centrality of a temporal vertex $\biv$ becomes negligible because most of the paths finish before or start after $\biv$ appears.
By contrast, if the time interval is too narrow, again the centrality of $\biv$ becomes negligible because paths can pass by only a tiny fraction of vertices in the time interval.
It should be noted that our centrality measures are free from any parameters not because we consider the centrality of temporal vertex. The centrality measures of a temporal vertex proposed in the previous work~\cite{Tang:2010,Tang:2010go,Kim:2012,Alsayed:2015,Pan:2011,Grindrod:2011,Estrada:2013,Grindrod:2013,Tang:2009,Rocha:2014,Motegi:2012,Nicosia:2013} require some parameters for different reasons.
Our centrality measures get around this issue by focusing on the local structure of temporal paths around the focal temporal vertex.
For the second issue, even if we compromise to use an approximation,
computing the approximated centrality value of a single temporal vertex requires computational time at least linear to network size~\cite{Bader:2007td}.

In this paper, we propose two novel centrality notions for temporal networks that resolve these issues.
The first one, called temporal coverage centrality (TCC), measures the fraction of pairs of (normal) vertices that have at least one fastest temporal path that uses the focal temporal vertex.
The second one, called temporal boundary coverage centrality (TBCC), measures the fraction of pairs of vertices that have a unique fastest temporal path, which uses the focal temporal vertex.

Our centrality notions address the two issues described above in the following way.
For the first issue, TCC and TBCC are free from setting of any parameters or time interval.
To calculate the TCC or TBCC value of a temporal vertex $\biv = (v,\tau)$,
we only have to run over all pairs of vertices $(u,w)$.
Namely, we consider temporal vertices $\biu = (u, \tau_u)$ and $\biw = (w,\tau_w)$,
where $\tau_u$ is the latest time at which we can send information from $u$ so that it reaches $v$ at time $\tau$,
and $\tau_w$ is the earliest time at which we can receive information at $w$ that is sent from $v$ at time $\tau$.
It should be noted that, if we fix focal temporal vertex $\biv$, $\tau_u$ and $\tau_w$ are uniquely determined by $u$ and $w$, respectively,
and that we thus do not have to care about the time interval around $\biv$.
Then, we check whether the information sent from $\biu = (u,\tau_u)$ to $\biw = (w,\tau_w)$ can or should drop by $\biv$.

For the second issue, although the definitions of TCC and TBCC might look complicated and hard to compute, this is not the case.
Indeed, computing TCC and TBCC can be reduced to the problem of deciding whether or not there is a directed path between queried vertices in an associated directed network (see Section~\ref{sec:dag} for details).
The latter problem is well studied in the database community \cite{simon1988improved,cohen2003reachability,yildirim2010grail,van2011memory,Yano:2013fq}, and it can be solved by constructing an index of the directed network, which computes the reachability between any pair of nodes by using information of the reachability between a fraction of node pairs.
If it suffices to use approximations to the TCC and TBCC values,
we only need to query the index at most ${\rm O}(\log^2 N)$ times, where $N$ is the total number of vertices in the network (see Appendix A). Since we can efficiently process queries to the index in practice, this method is advantageous compared to the ${\rm O}(N)$ time for approximating previous centrality notions.

With the aid of our centrality notions,
we are able to compute the centrality of all temporal vertices in a temporal network and analyze the statistics of the whole network.
Using TBCC, we demonstrate that real-world temporal networks have a small number of temporal vertices without which information propagates more slowly.
Surprisingly, we reveal that the temporal vertices of large centrality values form a narrow time region,
and this time region seemingly corresponds to the beginning or the end of a time interval in which temporal edges occur in a bursty manner.
In addition, by using TCC, we show that the remaining part of the temporal network is highly redundant in the sense that there are many ways to send information as quickly as possible.
Although these properties are recognized in the network science community \cite{Trajanovski:2012,Takaguchi:2012,Scellano:2013},
we quantitatively confirm it for the first time using our centrality notions.
We also demonstrate that the removal of temporal vertices according to their TBCC values is effective for hindering the propagation of information for both delaying and stopping it.

The paper is organized as follows.
In Section~\ref{sec:pre}, we introduce basic notions of temporal networks and the directed network associated with a temporal network.
Section~\ref{sec:centrality} introduces our centrality notions for temporal vertices,
and Section~\ref{sec:compute} explains detailed methods of computing our centrality notions.
Section~\ref{sec:result} is dedicated to demonstrating our experimental results.
We give the conclusion in Section~\ref{sec:conclusion}.

\section{Preliminaries about temporal networks}\label{sec:pre}
\subsection{Basic notions}

We introduce the terminology and symbols to describe temporal network structure, which basically follow those used in Ref.~\cite{Wu:2014}.
 
For integer $k$, let $[k]$ denote the set $\{1,2,\ldots,k \}$.
We define $\bbR_+$ as the set of non-negative real numbers.

Let $V$ be the set of vertices.
A temporal edge is represented by quadruplet $e = (u, v, \tau, \lambda)$, where $u, v \in V$, $\tau \in \bbR$, and $\lambda \in \bbR_+$.
For temporal edge $e = (u,v,\tau,\lambda)$,
we refer to $\tau$, $\lambda$, and $\tau + \lambda$ as the starting time, the duration, and the ending time of $e$, respectively.
Temporal network $G = (V, E)$ is a pair of set of vertices $V$ and set of temporal edges $E$.

When we study temporal networks, a vertex at a certain time is of interest.
Therefore, we define a temporal vertex by a pair of vertex $v \in V$ and time $\tau \in \bbR$.
In the following, we always use bold symbols such as $\biv$ to denote temporal vertices.
For temporal vertex $\biv = (v,\tau)$, we denote the time $\tau$ by $\tau(\biv)$.

Temporal path $P$ in temporal network $G=(V,E)$ is defined as an alternating sequence of temporal vertices and edges $P = \langle \biv_1,e_1,\biv_2,e_2,\ldots,e_{k-1},\biv_{k} \rangle$ satisfying the following properties.
Let $\biv_i = (v_i,\tau_i)$ for each $i \in [k]$.
Then for each $i \in [k-1]$,
the $i$-th temporal edge $e_i$ is of the form $e_i = (v_i,v_{i+1},\tau,\lambda)$ such that $\tau_i \leq \tau$ and $\tau + \lambda \leq \tau_{i+1}$.
We define the starting time, the duration, and the ending time of $P$ as $\tau_1$, $\tau_k - \tau_1$, and $\tau_k$, respectively.
For two temporal vertices $\biu$ and $\biv$,
relationship $\biu \leadsto \biv$ indicates that there is a temporal path from $\biu$ to $\biv$.

We define the earliest arrival time at vertex $w$ when departing from temporal vertex $\biv$ by the smallest $\tau \in \bbR$ such that $\biv \leadsto (w,\tau)$,
and we denote it by $\eat(\biv,w)$.
If there is no such $\tau$, we define $\eat(\biv,w) = \infty$.
Similarly, we define the latest departure time from a vertex $u$ for arriving at $\biv$ as the largest $\tau \in \bbR$ such that $(u,\tau) \leadsto \biv$,
and we denote it by $\ldt(\biv,u)$.
If there is no such $\tau$, we define $\ldt(\biv,u) = -\infty$.
A fastest temporal path from temporal vertex $\biv$ to vertex $w$ is a temporal path from $\biv$ to $(w,\eat(\biv,w))$,
and a fastest temporal path from a vertex $u$ to a temporal vertex $\biv$ is a temporal path from $(u,\ldt(\biv,u))$ to $\biv$.

\subsection{Directed acyclic graph representation}\label{sec:dag}
A directed acyclic graph (DAG) is a directed network with no directed cycle.
In this section, we describe the DAG representation of a temporal network,
which is useful when solving problems related to temporal paths and describing the centrality notions we will introduce in Section~\ref{sec:centrality}.
This DAG representation and its variants have been considered in the analysis of temporal networks~\cite{Kempe:2000,Kostakos:2009,Kim:2012,Pfitzer:2013,Scholtes:2014,Speidel:2015}.

For temporal network $G = (V, E)$, the DAG representation of $G$, denoted by $\widehat{G} = (\widehat{V},\widehat{E})$, is constructed as follows.
A vertex in $\widehat{G}$ represents a temporal vertex in $G$.
For each $v \in V$, we first add to $\widehat{V}$ two vertices corresponding to the temporal vertices $(v,-\infty)$ and $(v,\infty)$.
For each temporal edge $(u,v,\tau,\lambda) \in E$,
we add to $\widehat{V}$ two vertices corresponding to temporal vertices $\biu=(u,\tau)$ and $\biv=(v,\tau+\lambda)$ (if they do not exist in $\widehat{V}$) and add edge $(\biu,\biv)$ to $\widehat{E}$.
Finally, for each pair of temporal vertices $\biu = (u,\tau),\biu' = (u,\tau')$ sharing the same vertex $u$,
we add edge $(\biu,\biu')$ to $\widehat{E}$ if there is no temporal vertex of the form $(u,\tau'')$ in $\widehat{V}$ such that $\tau < \tau'' < \tau'$.

Figure~\ref{fig:dag-representation} illustrates DAG representation $\widehat{G}$ of temporal network $G$ shown in Fig.~\ref{fig:temporal-network}.
The vertex in the $i$-th row and the $j$-th column corresponds to the temporal vertex $(v_i,j)$.
For example, since there is temporal edge $(v_1,v_2,1,1)$ in $G$,
we have an edge from $(v_1,1)$ to $(v_2,2)$ in $\widehat{G}$.
For the $i$-th row, the leftmost and rightmost vertices correspond to the temporal vertices $(v_i,-\infty)$ and $(v_i,\infty)$, respectively.

\begin{figure}[!t]
\centering
\includegraphics[width = 0.8\hsize]{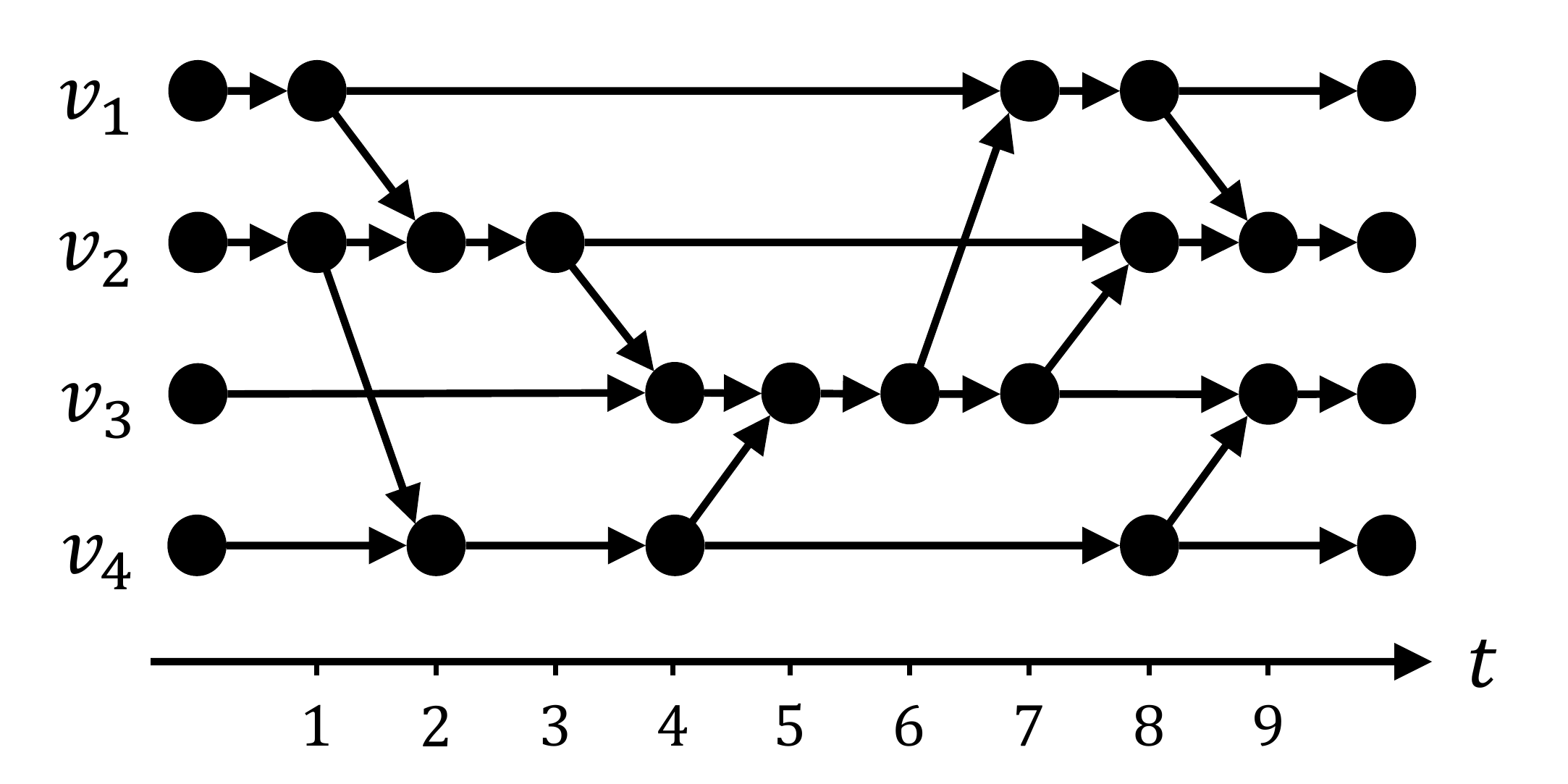}
\caption{DAG representation of the temporal network shown in Fig.~\ref{fig:temporal-network}.}
\label{fig:dag-representation}
\end{figure}

From the construction of the DAG representation, we have the following useful properties:
\begin{lemma}\label{lem:is-DAG}
  Let $G$ be a temporal network.
  Then, $\widehat{G}$ is a DAG.
\end{lemma}
\begin{proof}
  This is clear as we only add edges of the form $((u,\tau), (v,\tau'))$,
  where $\tau < \tau'$.
\end{proof}
\begin{lemma}\label{lem:DAG-representation}
  Let $G$ be a temporal network.
  Suppose that temporal vertices $\biu$ and $\biv$ have corresponding vertices in $\widehat{G}$.
  Then, there is a temporal path from $\biu$ to $\biv$ in $G$ if and only if there is a directed path from $\biu$ to $\biv$ in $\widehat{G}$.
\end{lemma}
\begin{proof}
  Let $P = \langle \biv_1,e_1,\biv_2,\ldots,e_{k-1},\biv_k \rangle$ be a temporal path from $\biv_1 = \biu$ to $\biv_k = \biv$.
 Without loss of generality, we assume that the time of $\biv_i$ is equal to the starting time of $e_i$ or the ending time of $\biv_{i-1}$. Then, each $\biv_i$ has a corresponding vertex in $\widehat{G}$.
  Let $\biv_i = (v_i,\tau^v_i)$ for each $i \in [k]$ and $e_i = (v_i,v_{i+1},\tau^e_i,\lambda^e_i)$ for each $i \in [k-1]$.
  Then, there is a directed path $(v_1,\tau^v_1), (v_1,\tau^e_1), (v_2,\tau^e_1+\lambda^e_1), (v_2,\tau^v_2), (v_2,\tau^e_1), (v_2,\tau^e_2 + \lambda^e_2),\ldots,(v_k,\tau^v_k)$ in $\widehat{G}$.
  The converse easily follows the correspondence explained above.
\end{proof}

\section{Temporal coverage centralities}\label{sec:centrality}

In this section, we introduce the temporal coverage centrality and the temporal boundary coverage centrality.

\subsection{Temporal coverage centrality}\label{sec:temporal-coverage-centrality}
Before defining TCC, we define the notion of coverage in temporal networks by generalizing its original version in static networks \cite{Yoshida:2014} as follows.
Let $\biv$ be a temporal vertex and $u,w$ be vertices.
Let $\biu = (u,\ldt(\biv,u))$ and $\biw = (w,\eat(\biv,w))$.
Then, we say that $\biv$ covers node pair $(u,w)$ if the following two conditions hold:
\begin{enumerate}
\item $\eat(\biu,w) = \eat(\biv,w)$,
\item $\ldt(\biw,u) = \ldt(\biv,u)$. 
\end{enumerate}
In words, the earliest arrival time at $w$ when departing from $\biu$ does not change even if we drop by $\biv$ (condition 1),
and the latest departure time from $u$ for arriving at $\biw$ does not change even if we drop by $\biv$ (condition 2).
Figure~\ref{fig:temporal-coverage-centrality} explains condition~1.
Let us focus on $\biv = (v_1,7)$.
Then, temporal vertices $\biu = (v_4,\ldt(\biv,v_4)) = (v_4,4)$ and $\biw = (v_2,\eat(\biv,v_2)) = (v_2,9)$ are determined as shown in the figure.
We observe that,
if we depart from $\biu$ and are not forced to drop by $\biv$, we can arrive at $\biw' = (v_2,8)$, which is earlier than $\biw$.
Hence, node pair $(u,w)$ is not covered by $\biv$ but by $\biw'$.

\begin{figure}[!t]
  \centering
  \includegraphics[width = 0.8\hsize]{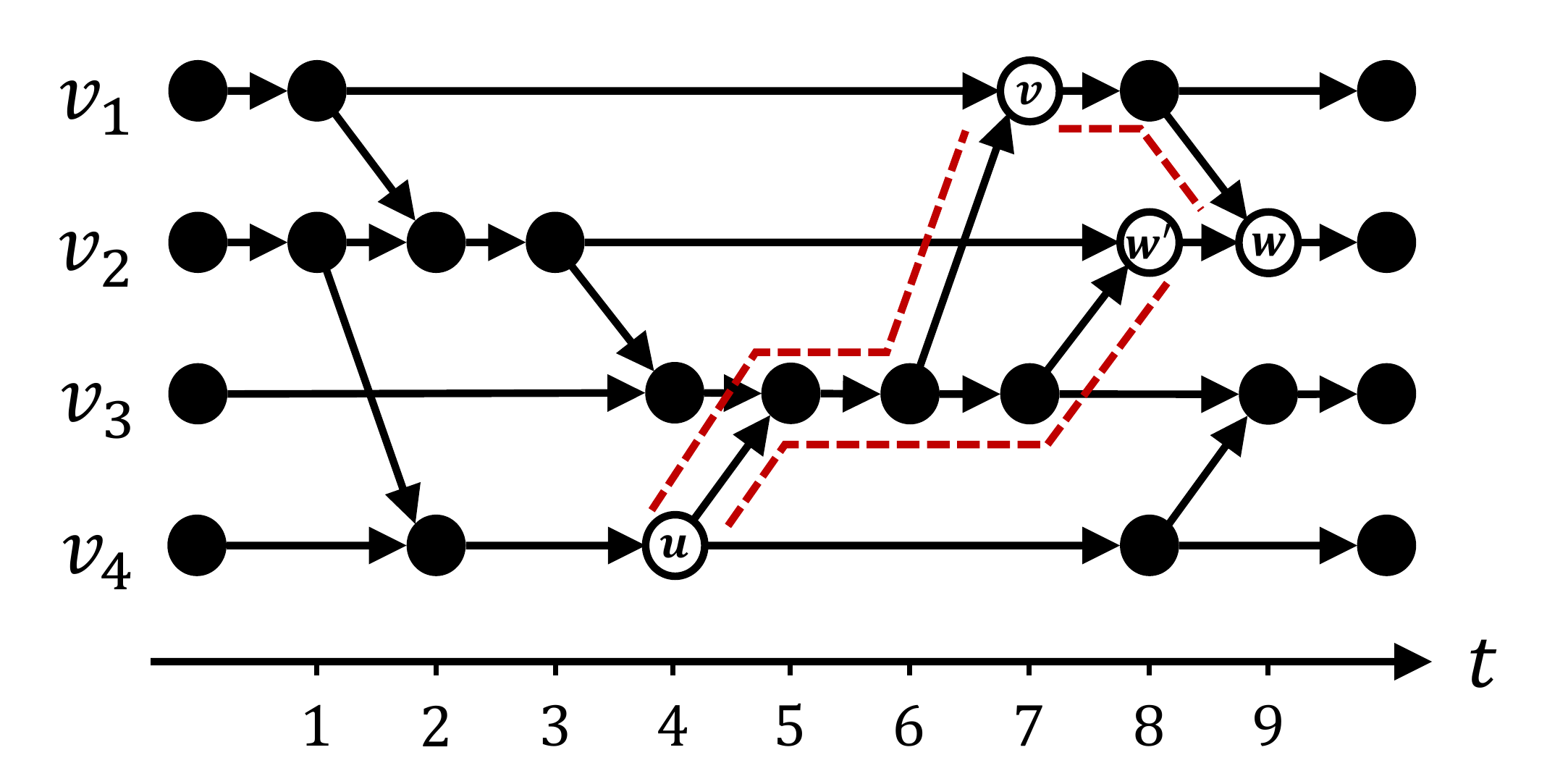}
  \caption{Schematic describing the concept of temporal coverage centrality.
  The dashed polygonal lines represent the two temporal paths from vertex $v_4$ to $v_2$ that contain temporal vertex $\biv$ in their durations.}
  \label{fig:temporal-coverage-centrality}
\end{figure}

On the basis of this notion of coverage, the TCC value of $\biv$ is defined as the fraction of pairs $(u,w) \in V \times V$ that are covered by $\biv$.
By definition, the TCC value of a temporal vertex takes a real number in $[0, 1]$. If the TCC value is close to unity, the temporal vertex is said to be central in the sense that it covers many pairs of nodes.
The formal definition is given in Algorithm~\ref{alg:temporal-coverage-centrality} in an algorithmic manner.

\begin{algorithm}[!t]
  \caption{(The TCC value of $\biv$)}
  \label{alg:temporal-coverage-centrality}
  \begin{algorithmic}[1]
    \STATE $r \leftarrow 0$.
    \FOR{$u \in V$ and $w \in V$}
      \STATE $\biu \leftarrow (u,\ldt(\biv,u))$.
      \STATE $\biw \leftarrow (w,\eat(\biv,w))$.
      \IF{$\eat(\biu, w) = \tau(\biw)$ and $\ldt(\biw,u) = \tau(\biu)$} \label{line:coverage-check}
        \STATE $r \leftarrow r + 1$.
      \ENDIF
    \ENDFOR
    \STATE \textbf{return} $r / |V|^2$.
  \end{algorithmic}
\end{algorithm}

\subsection{Temporal boundary coverage centrality}\label{sec:temporal-boundary-coverage-centrality}

Let $\biv = (v,\tau)$ be a temporal vertex and $u,w$ be vertices.
Let $\biu = (u,\ldt(\biv,u))$ and $\biw = (w,\eat(\biv,w))$.
Even if the TCC value of $\biv$ is large,
it does not always imply that removing the temporal edges involving $\biv$ makes $\eat(\biu,w)$ larger or $\ldt(\biw,u)$ smaller.
One particular reason for this is that sometimes we can reach $v$ from $\biu$ earlier than $\tau$ and can leave $v$ later than $\tau$ to reach $\biw$ (see temporal vertices $\biv_2$ and $\biv_3$ in Fig.~\ref{fig:temporal-boundary-coverage-centrality}).
In some applications, we may want to regard such $\biv$ as unimportant.

To address this issue, we define TBCC by imposing additional criteria to the notion of coverage as follows.
Note that, if focal temporal vertex $\biv$ is an example of the situation stated in the previous paragraph,
then $\eat(\biu,v) < \tau$ or $\ldt(\biw,v) > \tau$ should hold.
Hence, we define that a pair $(u,w)$ of vertices is covered at a boundary by temporal vertex $\biv$ if the following hold:
\begin{enumerate}
\item $(u,w)$ is covered by $\biv$, and
\item $\eat(\biu,v) = \tau$ or $\ldt(\biw,v) = \tau$.
\end{enumerate}

We explain this definition using the example shown in Fig.~\ref{fig:temporal-boundary-coverage-centrality}.
Let $\biv_i = (v,\tau_i)$ for $i \in [4]$.
Note that all $\biv_i$ $(i \in [4])$ cover vertex pair (u, w) as $\biu = (u,\ldt(\biv_i, u))$ and $\biw = (w,\eat(\biv_i,w))$ hold for all $i \in [4]$.
In addition, note that all $\biv_i$ cover $(u,w)$. We can see that $\biv_1$ and $\biv_4$ cover $(u,w)$ at the boundary because $\eat(\biu,v) = \tau_1$ and $\ldt(\biw,v) = \tau_4$.
By contrast, $\biv_2$ and $\biv_3$ do not cover $(u,w)$ at the boundary.

On the basis of this notion of coverage at the boundary, the TBCC value of $\biv$ is defined as the fraction of pairs $(u,w)$ that are covered at the boundary by $\biv$.
Similar to TCC, the TBCC value of a temporal vertex takes a real number in $[0, 1]$ by definition.
The formal definition is given in Algorithm~\ref{alg:temporal-boundary-coverage-centrality} in an algorithmic manner.

In closing this section, it should be noted the difference between the previous notion of the temporal betweenness centrality and TCC (and TBCC). 
The main difference lies in the normalization of the number of vertex pairs covered by the temporal vertex.
The definitions of TCC and TBCC do not normalize the number of such vertex pairs with the number of the fastest temporal paths, whereas the previous temporal betweenness centrality divides the number of the fastest paths that use the focal temporal vertex by the total number of the fastest temporal paths in the focal time window, as the betweenness centrality for static networks does~\cite{Tang:2010,Tang:2010go,Kim:2012,Alsayed:2015}.
We took such definitions of TCC and TBCC for the following reasons.
First, TCC and TBCC become free from any parameters because we do not need to set the time window to count the number of the relevant fastest temporal paths for the normalization. Second, the TCC and TBCC values are easy to interpret as the fraction of the vertex pairs that have a fastest temporal path using the focal temporal vertex.

\begin{figure}[!t]
  \centering
  \includegraphics[width = 0.75\hsize]{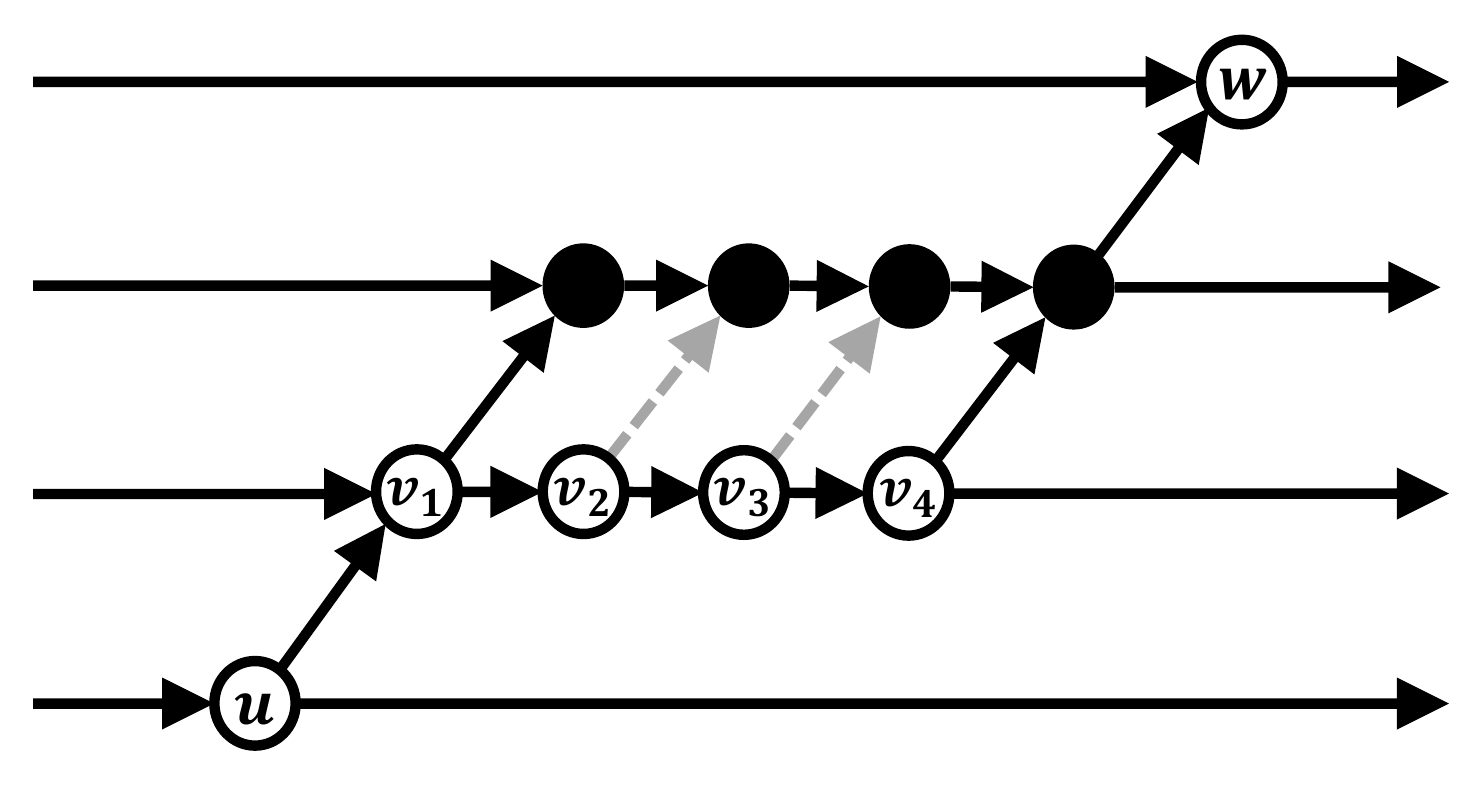}
  \caption{Schematic describing the concept of temporal boundary coverage centrality.
  The dashed arrows represent the temporal edges that do not contribute the centrality values of the source temporal vertices.}
  \label{fig:temporal-boundary-coverage-centrality}
\end{figure}

\begin{algorithm}[!t]
  \caption{(The TBCC value of $\biv$)}
  \label{alg:temporal-boundary-coverage-centrality}
  \begin{algorithmic}[1]
    \STATE $r \leftarrow 0$.
    \FOR{$u \in V$ and $w \in V$}
      \STATE $\biu \leftarrow (u,\ldt(\biv,u))$.
      \STATE $\biw \leftarrow (w,\eat(\biv,w))$.
      \IF{$\eat(\biu, w) = \tau(\biw)$ and $\ldt(\biw,u) = \tau(\biu)$}
        \IF{$\eat(\biu,v) = \tau(\biv)$ or $\ldt(\biw,v) = \tau(\biv)$} \label{line:coverage-at-boundary-check}
          \STATE $r \leftarrow r + 1$.
        \ENDIF
      \ENDIF
    \ENDFOR
    \textbf{return} $r / |V|^2$.
  \end{algorithmic}
\end{algorithm}

\section{Computing temporal coverage centralities}\label{sec:compute}

We can straightforwardly calculate TCC and TBCC according to Algorithms~\ref{alg:temporal-coverage-centrality} and~\ref{alg:temporal-boundary-coverage-centrality}.
In this section, to manage large temporal networks, we give efficient methods for computing TCC and TBCC on the basis of a graph--indexing technique developed recently in the database community~\cite{Yu:2010}, in particular, the method proposed in \cite{Yano:2013fq}.
The key idea is in how to speed up the computation of $\eat$ and $\ldt$ in Algorithms~\ref{alg:temporal-coverage-centrality} and~\ref{alg:temporal-boundary-coverage-centrality}.
We describe the exact computation of TCC and TBCC in this section, and we also give the algorithms to approximate the TCC and TBCC values whose running time is polylogarithmic in the total number of vertices in $G$ (see Appendix B). 

In a directed network, we say that a vertex $v_t$ is reachable from $v_s$ if there is a directed path from $v_s$ to $v_t$.
With respect to Lemma~\ref{lem:DAG-representation}, to enumerate the number of pairs $(u, w)$ being covered by $\biv$ (at the boundary, if needed), we want to efficiently answer reachability in the DAG representation $\widehat{G}$ of given temporal network $G$.
To this end, it is beneficial to construct an index of $\widehat{G}$ that computes the reachability between any pair of nodes on the basis of information of the reachability between a fraction of node pairs.
Such an index is often called a reachability oracle in the database community~\cite{simon1988improved,cohen2003reachability,yildirim2010grail,van2011memory,Yano:2013fq}.

The basic idea of the construction of a reachability oracle for the present problem is the following.
Naively, we want to compute a large table that stores the reachability of every pair of temporal vertices.
If this were possible, we could answer reachability just by looking at that table.
Unfortunately, however, perfecting this table requires $O(|\widehat{V}|^2)$ computation time and $O(|\widehat{V}|^2)$ space, which could be prohibitively slow and large.
The reachability oracle overcomes this problem by carefully storing partial information of the network.
Based on the information, it efficiently computes the reachability for the whole network.

The method proposed in Ref.~\cite{Yano:2013fq}, which we will use for the numerical experiments in Section \ref{sec:result}, computes a small table for each temporal vertex that stores reachability from (and to) a smaller number of other certain temporal vertices than the number of all the temporal vertices. It depends on the structure of each temporal network how small the table becomes.
Then, we can answer the reachability from a temporal vertex $\biu$ to a temporal vertex $\biv$ by checking whether there is another temporal vertex $\biw$ such that we can confirm the reachability from $\biu$ to $\biw$ and from $\biw$ to $\biv$ using the small tables of $\biu$ and $\biv$.
If there is such $\biw$, we indeed have a directed path from $\biu$ to $\biv$.
The challenging part of the construction lies in guaranteeing the other direction; if there is a directed path from $\biu$ to $\biv$, then there is always such $\biw$.
In addition, we need to be able to compute the small table for each vertex efficiently.
This method resolves these issues, so that it can handle directed networks of millions of edges with the query time of less than a microsecond on average (see Ref.~\cite{Yano:2013fq} for further technical details).

\section{Results}\label{sec:result}


The basic statistics of the datasets we use are summarized in Table~\ref{tbl:dataset}.
It should be noted that we do not use the actual time stamps in the datasets but define $\tau$ by the order of unique values of the time stamps. For example, if the dataset consists of two time stamps $t = 1,4$, we translate them into $\tau = 1,2$.
In addition, we assume that $\lambda$ is equal to the finest time resolution of each dataset for all the temporal edges.
Although interactions in Irvine and Email are directed (i.e., from sender to receiver(s) of messages), we regard them as undirected.

\begin{table}
  \centering
  \caption{Basic statistics of the datasets. Variables $n$, $m$, $\widehat{n}$, and $\tau_{\max}$ are the total number of vertices and temporal edges in $G$, the total number of vertices in $\widehat{G}$, and the maximum ending time of a temporal edge, respectively. The datasets are arranged in increasing order of $m$.}
  \label{tbl:dataset}
  \begin{tabular}{|c|r|r|r|r|}
    \hline
    Name & $n$ & $m$ & $\widehat{n}$ & $\tau_{\max}$\\
    \hline
    \hline
    Infectious~\cite{Isella:2011} & 410 & 17298 & 32218 & 1393\\
    \hline
    HT09~\cite{Isella:2011} & 113 &  20187 & 48477 & 5246\\
    \hline
    Hospital~\cite{Vanhems:2013} & 75 & 32424 & 65296 & 9454\\
    \hline
    Irvine~\cite{Opsahl:2009} & 1899 & 59835 & 220772 & 58192\\
    \hline
    Email~\cite{Michalski:2011} & 167 & 82927 & 254533 & 57843\\
    \hline
  \end{tabular}
\end{table}

\subsection{Statistics of TCC and TBCC}

Figure~\ref{fig:plot-TCC-and-TBCC} depicts the rank plots of the TCC and TBCC values of temporal vertices in the decreasing order.
In all the datasets except for the Email data, at least $10\%$ of temporal vertices have TCC values larger than $0.1$ (Fig.~\ref{fig:plot-TCC-and-TBCC}(a)).
This fact implies the redundancy of temporal networks in the sense that, when information flows between temporal vertices,
it can drop by different vertices without increasing the total duration of the temporal paths.
However, there are a smaller number of temporal vertices with large TBCC values (Fig.~\ref{fig:plot-TCC-and-TBCC}(b)).
This fact also implies the redundancy of temporal networks in a different sense such that, when information flows between temporal vertices,
it is not forced to exist at a certain vertex at a certain time.

\begin{figure}[!t]
\centering
\includegraphics[width=0.49\hsize]{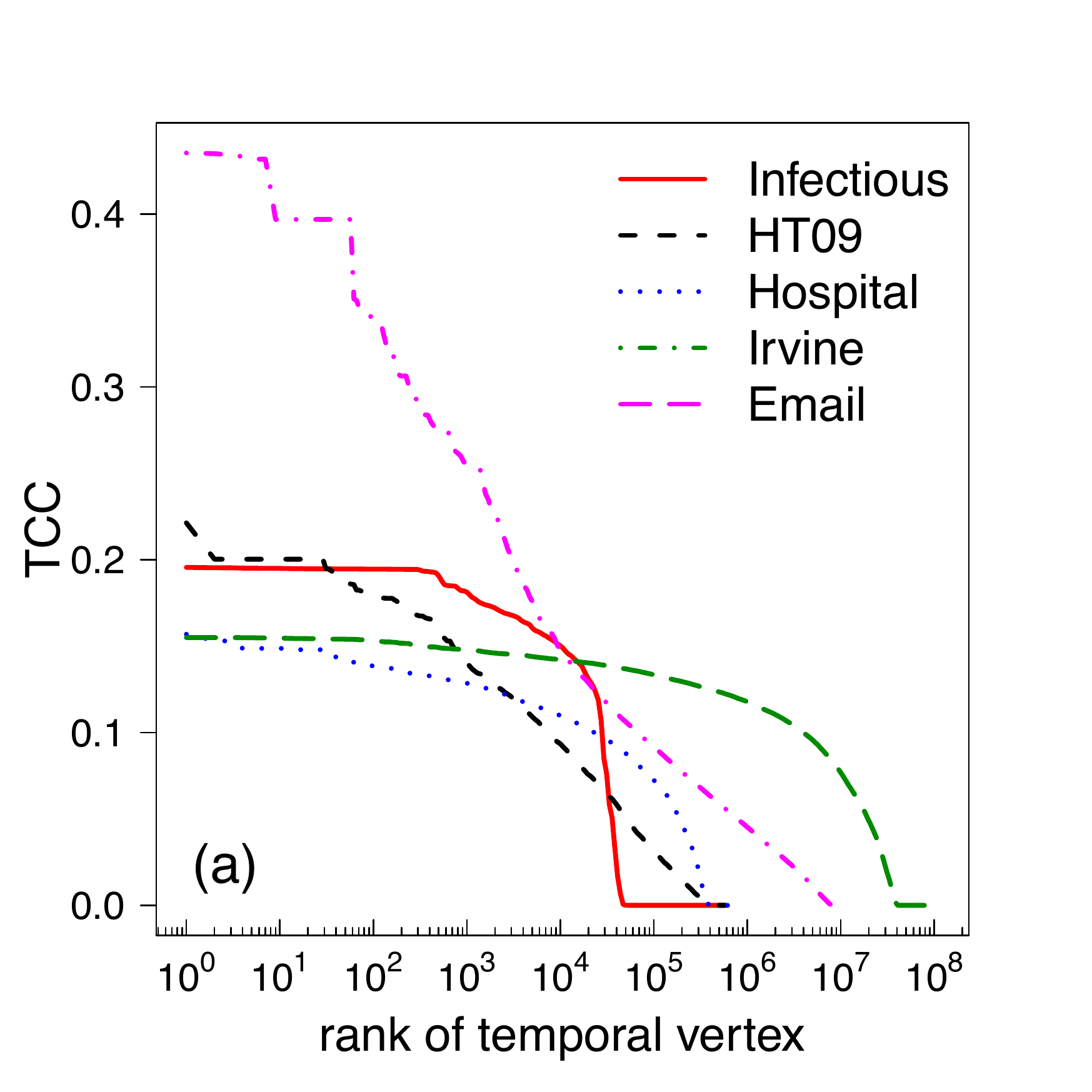}
\includegraphics[width=0.49\hsize]{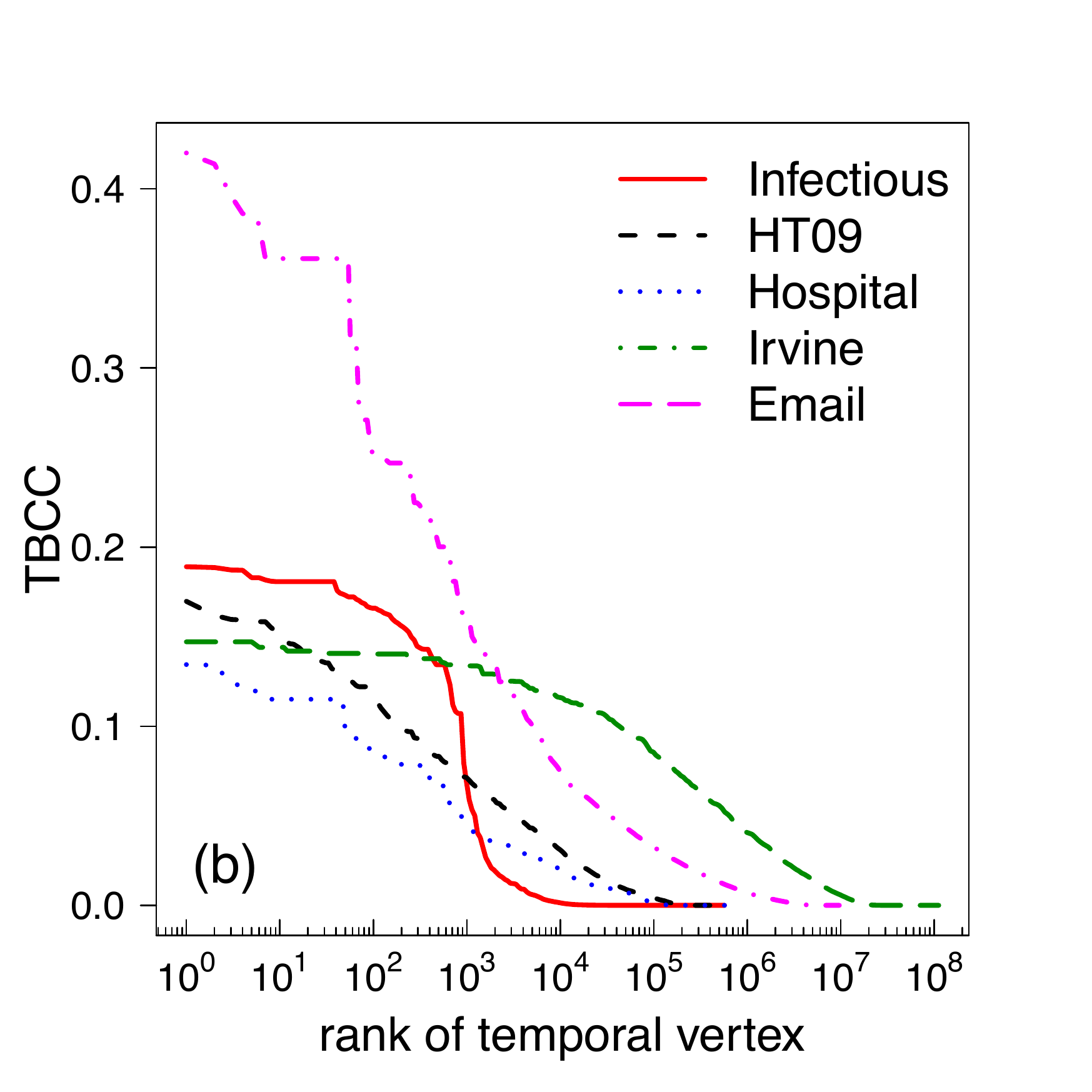}
\caption{Rank plots of the (a) TCC and (b) TBCC values.}
\label{fig:plot-TCC-and-TBCC}
\end{figure}

To see the impact of the structural peculiarity of temporal networks on these distributions, we computed the centrality values of temporal vertices in randomized temporal networks.
We randomize an original temporal network by replacing the two ends of each temporal edge by vertices chosen uniformly at random (similar to the procedure called randomized edges with randomly permuted times in Ref.~\cite{Holme:2012}).
The resultant centrality values are shown in Fig.~\ref{fig:plot-TCC-and-TBCC-RG}.
We notice that more temporal vertices have sufficiently large centrality values (e.g., larger than $0.1$) in real-world temporal networks (Fig.~\ref{fig:plot-TCC-and-TBCC}) than in randomized temporal networks (Fig.~\ref{fig:plot-TCC-and-TBCC-RG}).
The maximum centrality values are larger in the randomized than in the original networks for HT09 and Hospital, and vice versa for Infectious and Email.
This fact implies that the way the flow concentrates upon temporal vertices depends on each dataset.

It should be noted that the calculation for the randomized Irvine dataset did not stop even though the Email dataset, which has larger $\hat{n}$ than the Irvine, stopped. We can explain this result with the increase in the number of vertex pairs connected via temporal paths.
The dominant factor of the computational time is the number of vertex connected via temporal paths because we have to consider all of such vertex pairs to calculate the centrality value of a temporal vertex. After the randomization, most of the vertex pairs are likely to have temporal paths and the number of such pairs scales with $n^2$. If we take into account that the Irvine dataset has the largest $n$ value among the five datasets we consider, it makes sense for the Irvine dataset to require the far longer computational time compared to the other four datasets. 

\begin{figure}[!t]
\centering
\includegraphics[width=0.49\hsize]{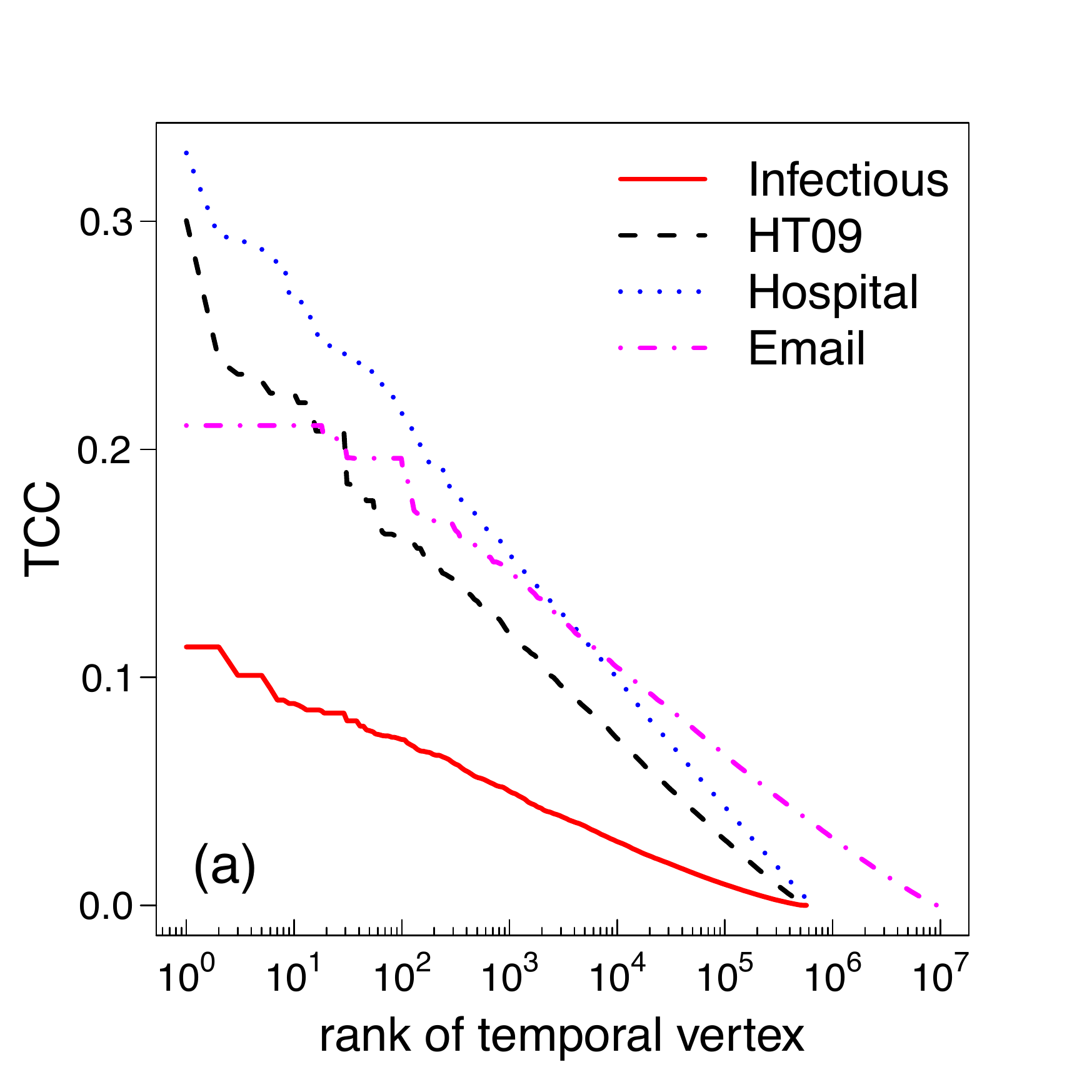}
\includegraphics[width=0.49\hsize]{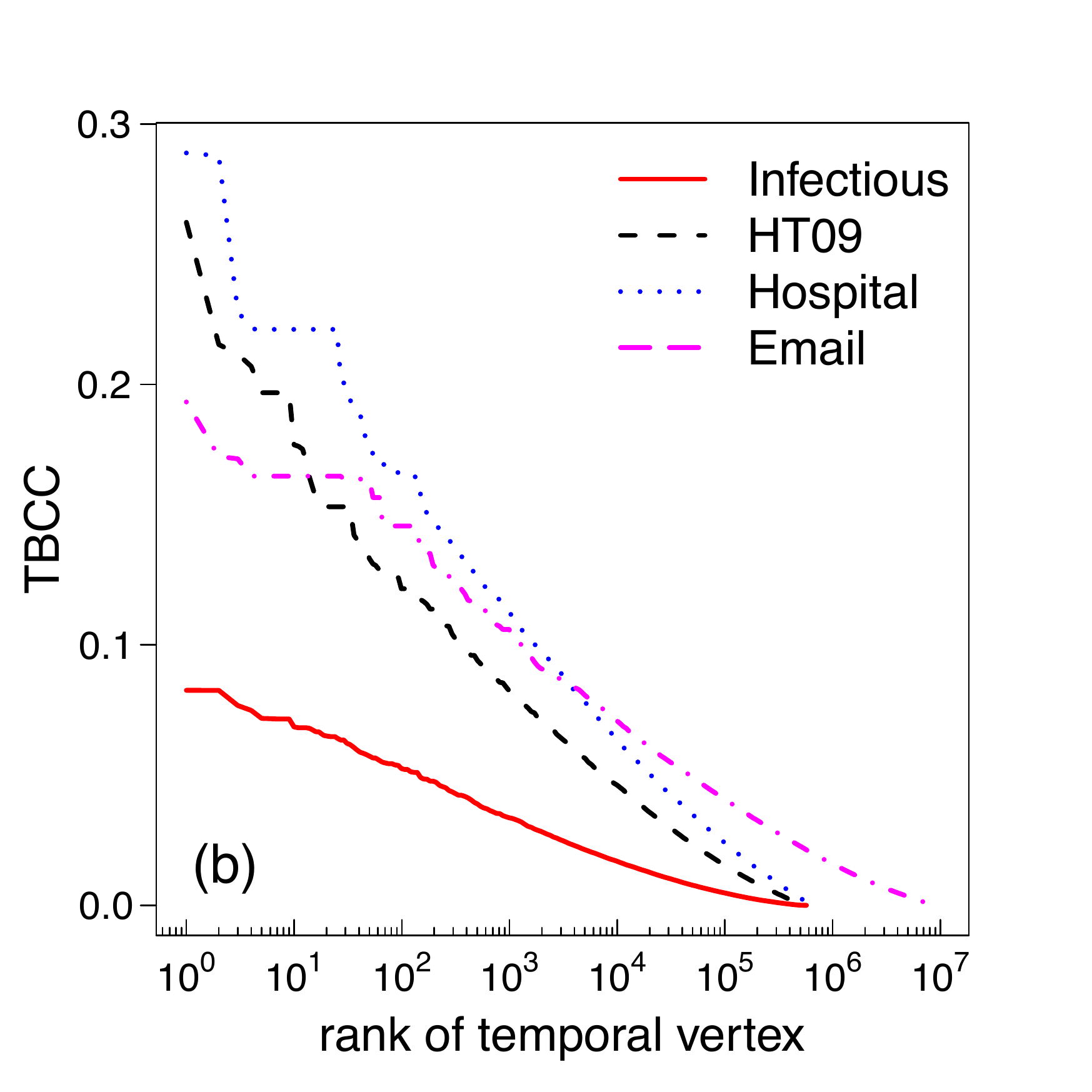}
\caption{Rank plots of the (a) TCC and (b) TBCC values in randomized temporal networks.
The curves for Irvine are not provided because the computation did not stop.}
\label{fig:plot-TCC-and-TBCC-RG}
\end{figure}

Next, we examine how the centrality values change over time owing to the structural transformation of the temporal networks.
Figure~\ref{fig:transition-of-max-centrality} depicts the change in the maximum TCC and TBCC values over temporal vertices at present and the number of temporal vertices at present for Infectious and Hospital.
In both datasets shown in Fig.~\ref{fig:transition-of-max-centrality}, we can see some periodic patterns in the number of temporal vertices.
However, the maximum centrality values are not much affected by the patterns, which implies that these values are determined not by the mere activity level in the networks but by the structure of the temporal network.
In addition, the fact that the maximum centrality values vary considerably throughout the observation periods suggests that we should carefully incorporate temporal structure to assess the importance of vertices.
Generally, the maximum TCC values are larger than the maximum TBCC values, which makes sense according to their definitions (i.e., TBCC only counts the coverage of the temporal paths at the boundary but TCC does not impose this boundary criterion).

\begin{figure*}[!t]
\centering
\includegraphics[width=0.48\hsize]{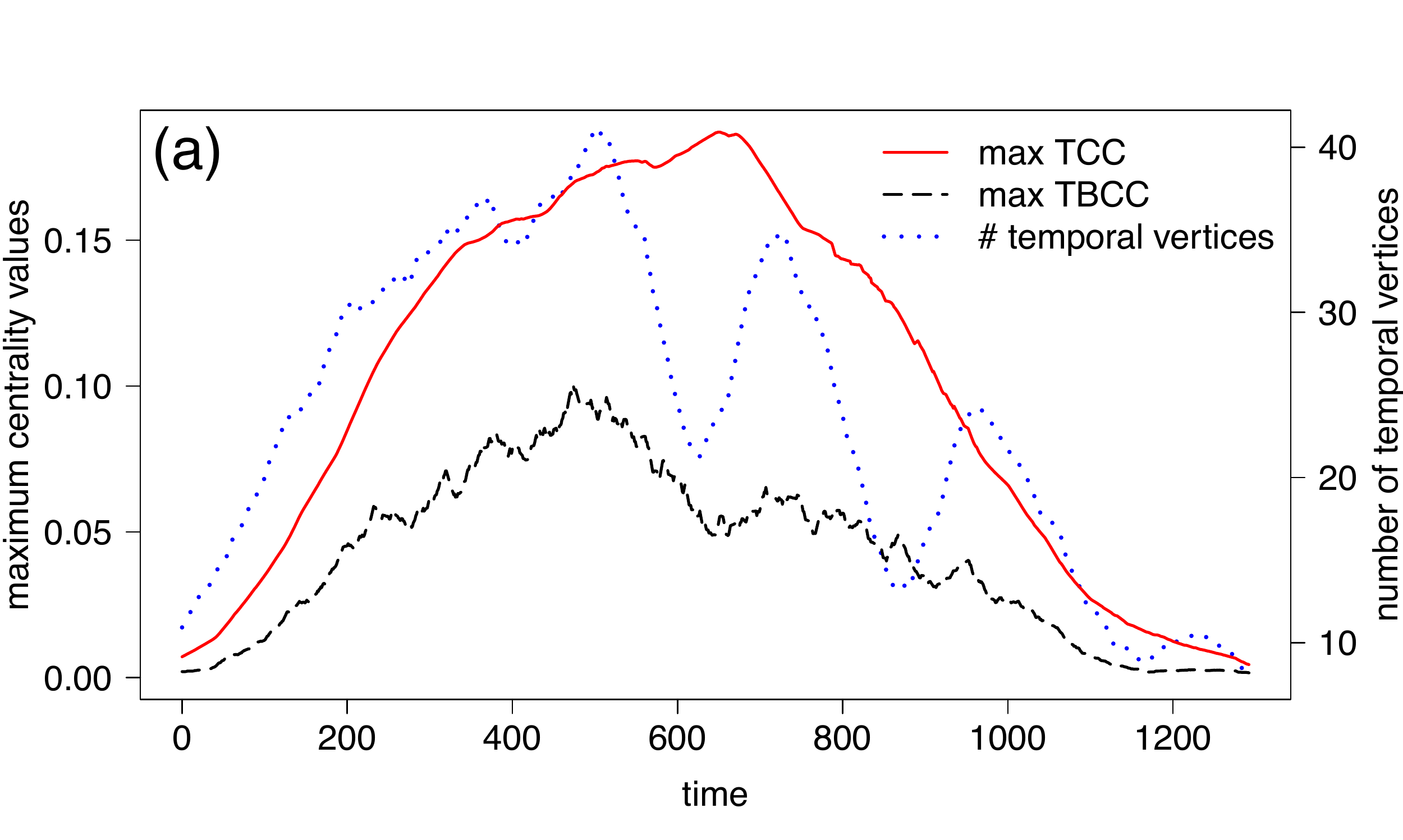}
\includegraphics[width=0.48\hsize]{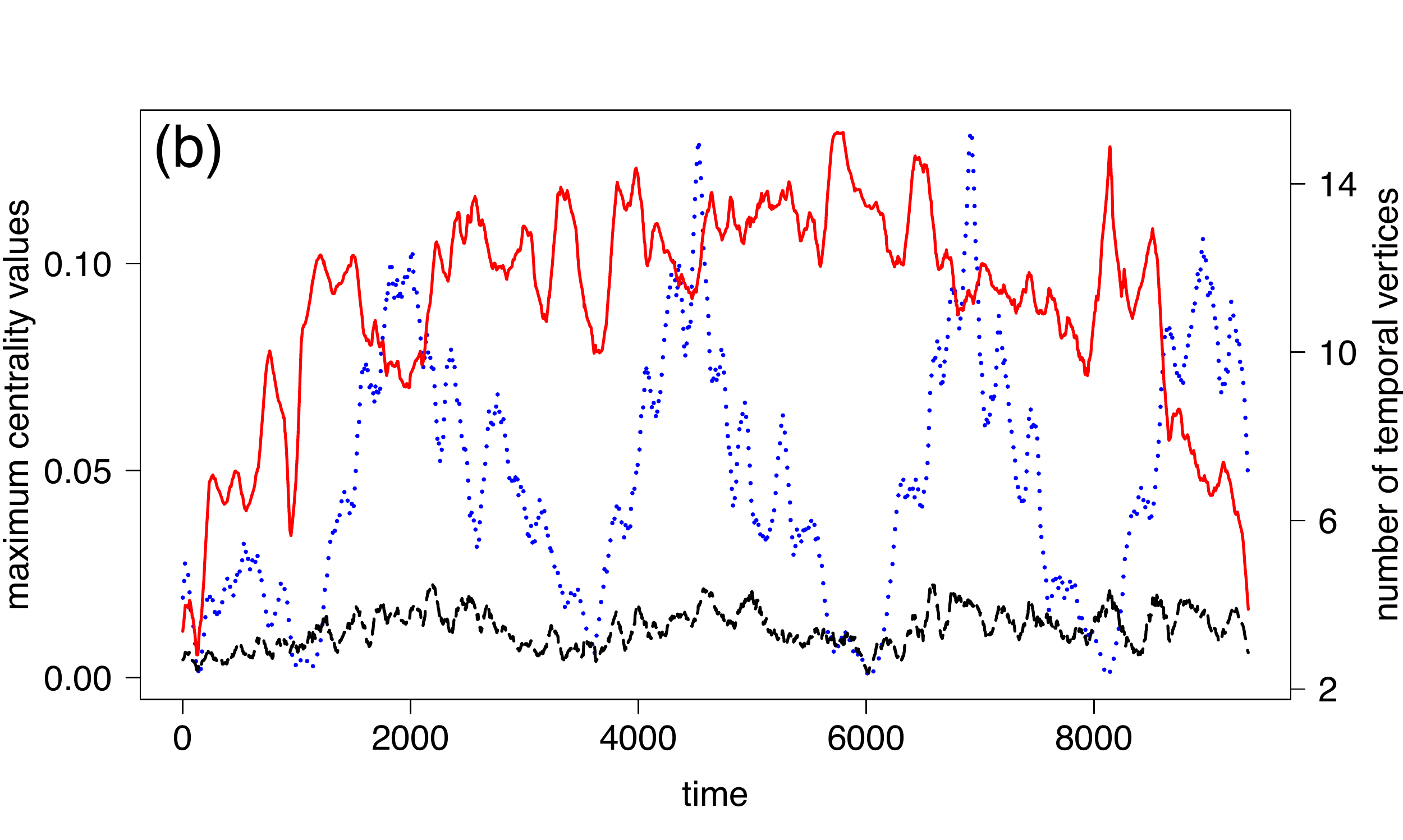}
\caption{Change in the maximum TCC and TBCC values over temporal vertices at present in (a) Infectious and (b) Hospital.
For readability, we smoothed the curves by taking the average over a sliding window with a length of $100$ units of time.
}
\label{fig:transition-of-max-centrality}
\end{figure*}

When we focus on a particular vertex, two centrality values of it also vary in a different manner over time.
Figure~\ref{fig:transition-of-single-vertex} depicts the change in the TCC and TBCC values of the vertex that are involved in the largest number of temporal edges in the two datasets, Infectious and Hospital.
The TCC value of the vertex increases with time in Infectious (Fig.~\ref{fig:transition-of-single-vertex}(a)), simply because the number of present temporal vertices increases and thus the focal vertex can reach these vertices in this period (also see Fig.~\ref{fig:transition-of-max-centrality}(a)).
By contrast, the TBCC value does not exhibit such an increasing trend. This fact supports our original purpose of introducing TBCC, i.e., to discount the centrality values of the temporal vertices of the dispensable temporal paths.
In addition, the plot of TBCC unveils that even the vertex with the largest number of temporal edges does not always bridge effective temporal paths.
In Hospital (Fig.~\ref{fig:transition-of-single-vertex}(b)), we can observe that the temporal edges associated with the focal vertex are partitioned into five time intervals, in each of which temporal edges occur in a bursty manner, and the centrality values of the vertex become larger at the beginning and the end of each of these time intervals.
This observation makes sense because, at the endpoints of a time interval, a vertex tends to play the role as the gateway for information flowing into or out of the time interval.

\begin{figure*}[!t]
\centering
\includegraphics[width=0.48\hsize]{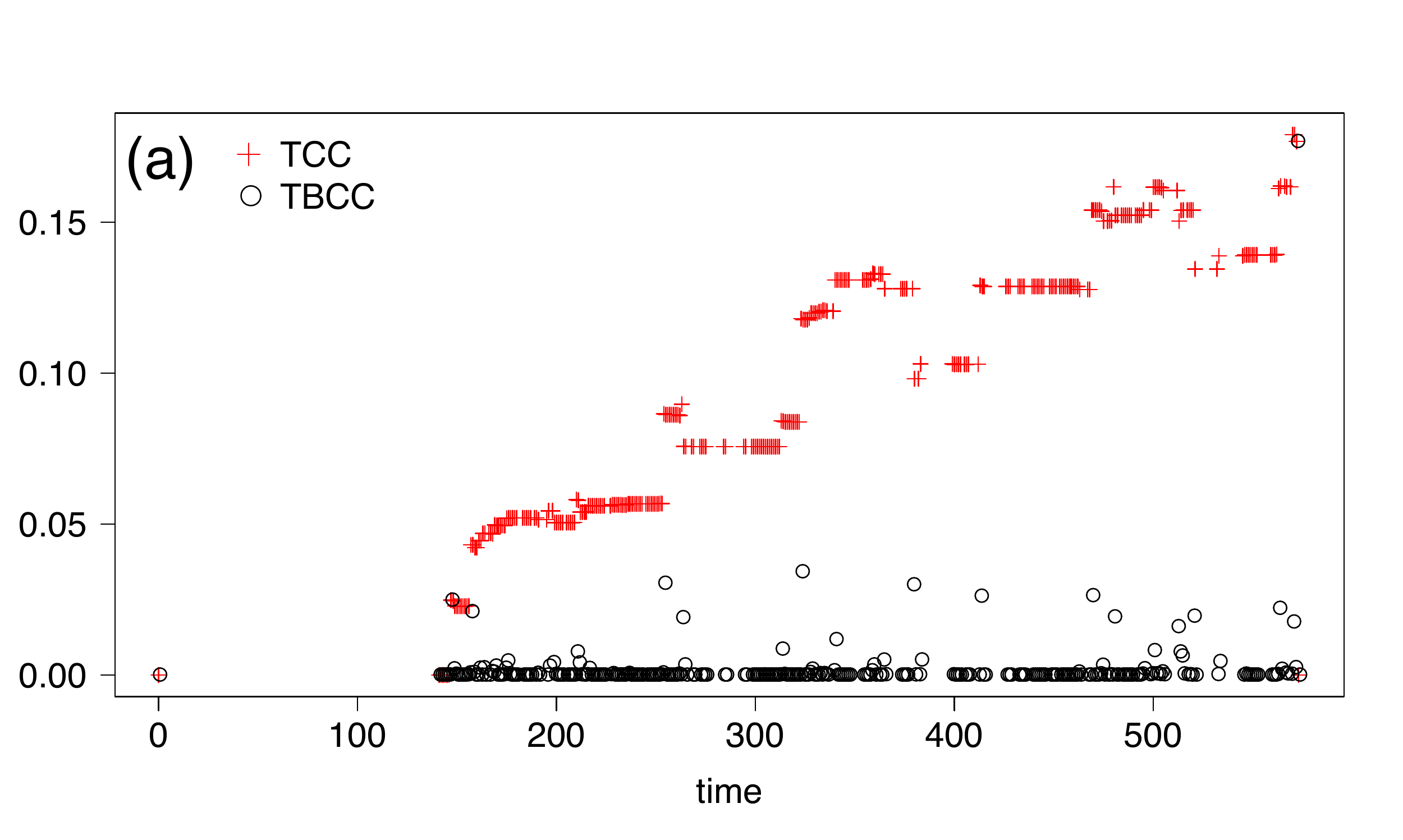}
\includegraphics[width=0.48\hsize]{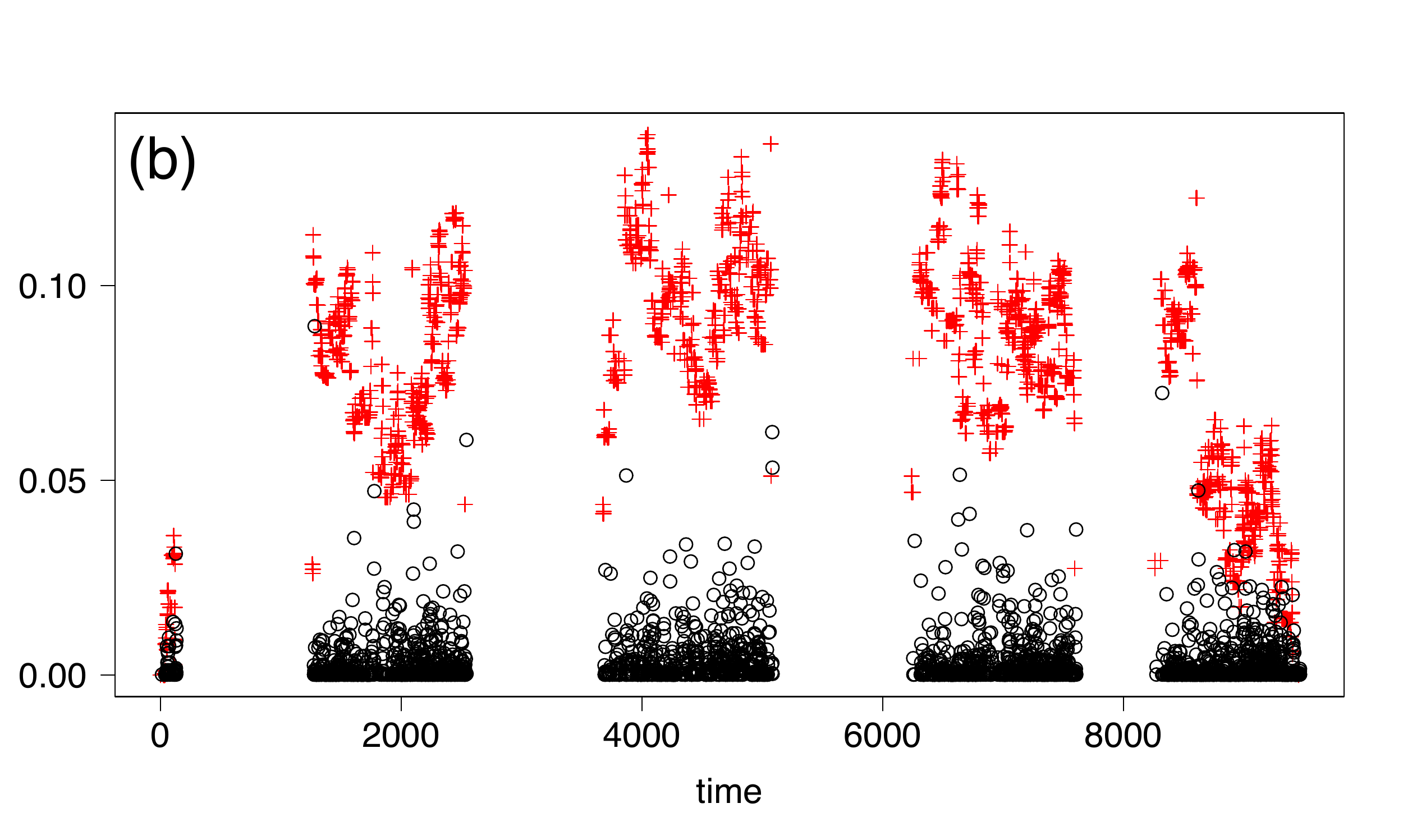}
\caption{Change in the TCC and TBCC values of the vertex with the largest number of temporal edges. (a) Vertex with label~$195$ in Infectious and (b) vertex with label $1115$ in Hospital.}
\label{fig:transition-of-single-vertex}
\end{figure*}

The computational efficiency of the two centralities enables us to draw a map of the centrality values of all the temporal vertices over time.
This map reveals the existence of bottleneck time regions in the empirical temporal networks.
Figures~\ref{fig:heatmap-TCC}(a) and \ref{fig:heatmap-TCC}(b) depict the TCC values of temporal vertices as a heat map for Infectious and Hospital, respectively.
In both datasets, most temporal vertices have non-negligible TCC values, and these results support the notion of redundancy of temporal networks (see Fig.~\ref{fig:plot-TCC-and-TBCC}(a)) such that all the vertices can belong to redundant temporal paths.
In addition, the temporal vertices with the largest centrality values appear in the middle of the observation period (around time $700$ and $6000$ in Figs.~\ref{fig:heatmap-TCC}(a) and (b), respectively), and the temporal vertices at the same time tend to have similar TCC values.
We found the same phenomenon in all the datasets (see Electronic Supplementary Materials for the plots of the other datasets),
and the existence of this bottleneck time period seems to be a common property of empirical temporal networks.

If we are interested in when these bottleneck time periods begin and end, we can look at the heat map of the TBCC values. 
As an example, Fig.~\ref{fig:heatmap-TCC}(c) magnifies a bottleneck time period in Infectious (Fig.~\ref{fig:heatmap-TCC}(a)) in which we observe many temporal vertices with the largest TCC values. However, the boundary of the bottleneck period is not clear in the figure.
Figure~\ref{fig:heatmap-TCC}(d) shows the heat map of the TBCC values in the same area as shown in Fig.~\ref{fig:heatmap-TCC}(c).
As we observe, the TBCC values indicate the boundaries at $\tau \simeq 660$, $680$, and $750$. 
This boundary information should be meaningful, for example, when we narrow the candidates of the vertices to be vaccinated for epidemic spreading on temporal networks~\cite{Lee:2012,Starnini:2013,Masuda:2013}.

We finally stress again that it becomes possible to compute these statistics and analyze the structure of temporal networks in such detail because of the efficient computation of TCC and TBCC using the reachability oracle.

\begin{figure*}
\centering
\includegraphics[width=0.48\hsize]{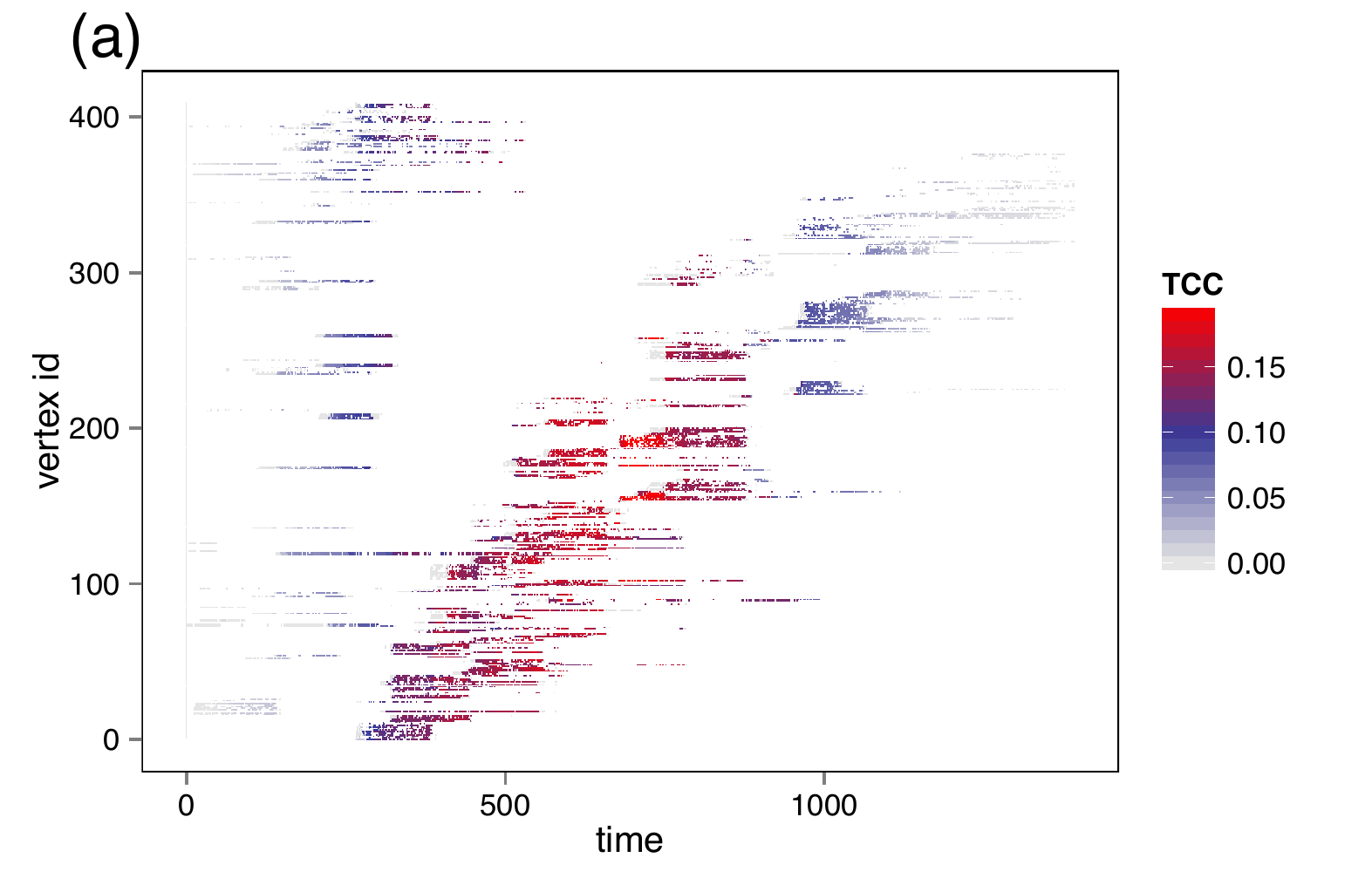}
\includegraphics[width=0.48\hsize]{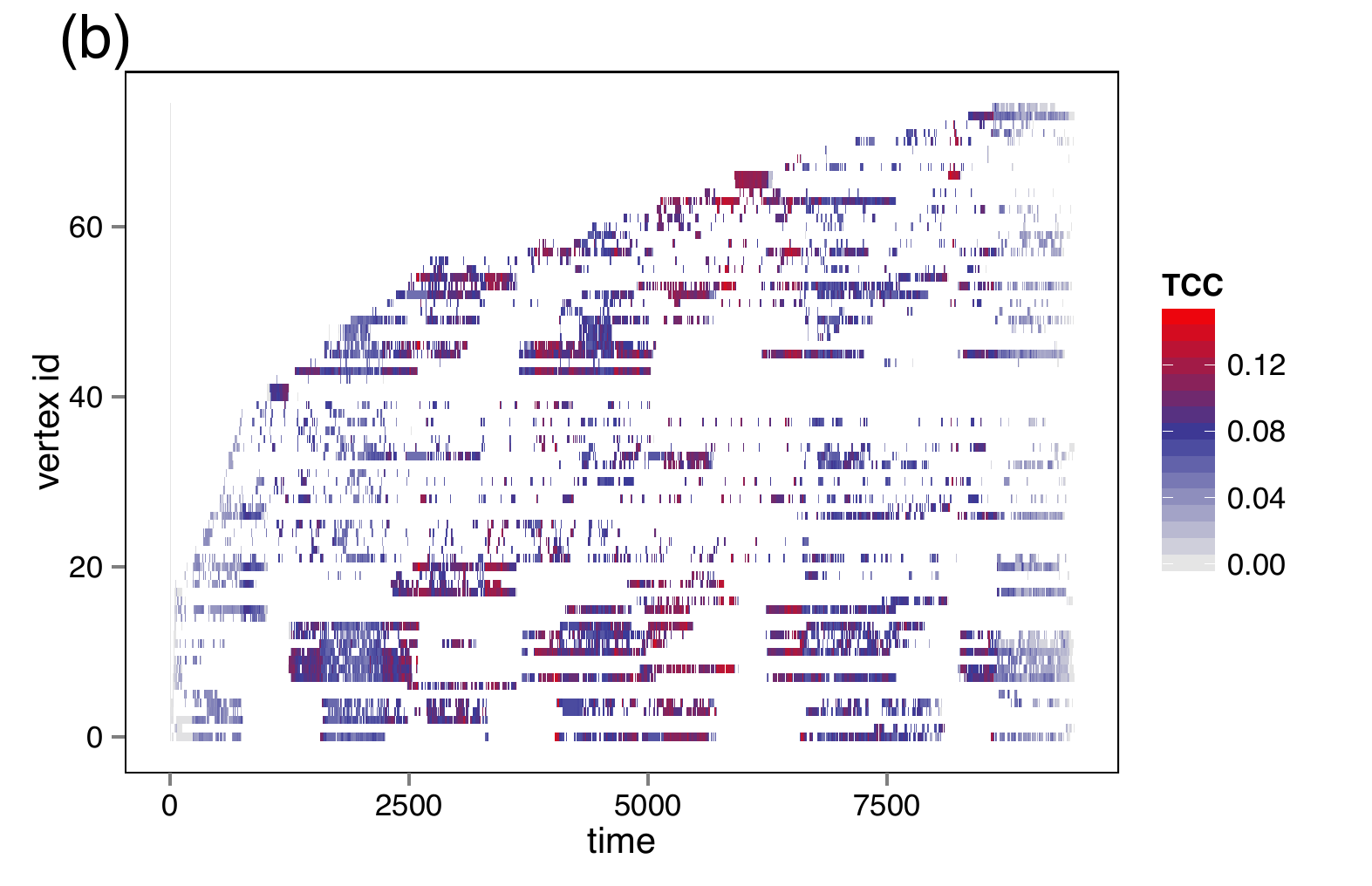}
\includegraphics[width=0.48\hsize]{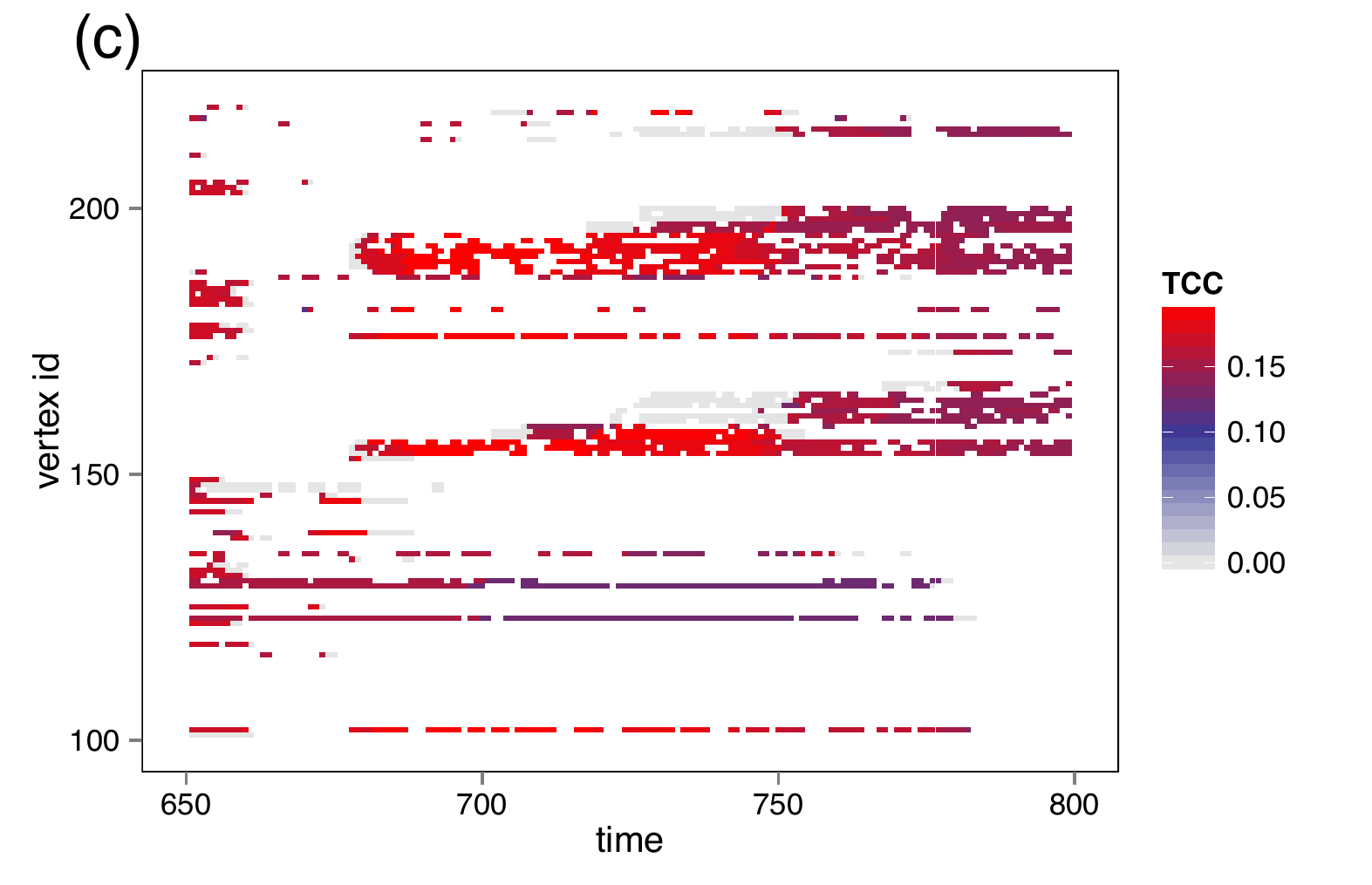}
\includegraphics[width=0.48\hsize]{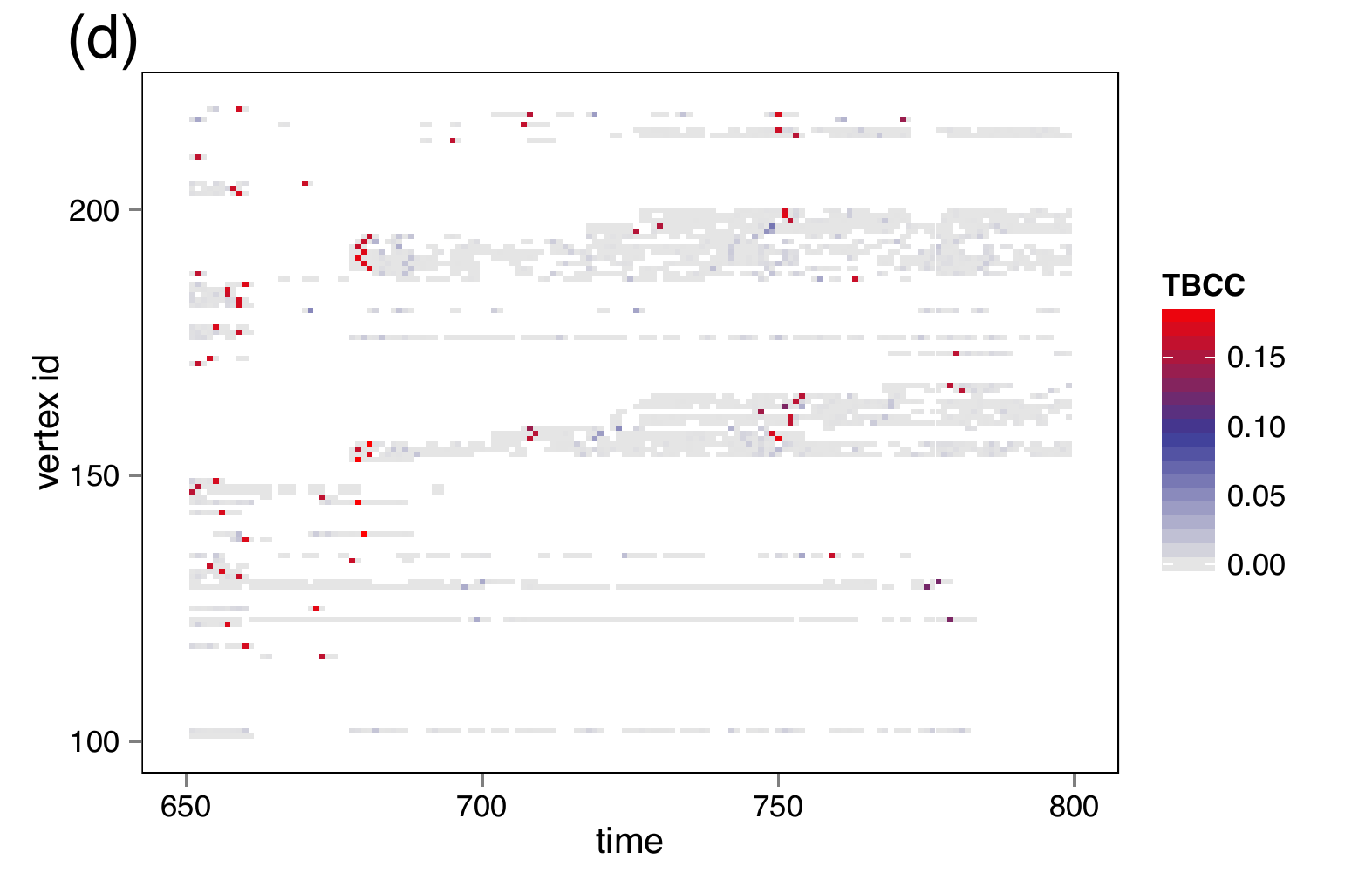}
\caption{Heat maps of the TCC values for (a) Infectious and (b) Hospital.
(c) Heat map magnifying the area with $650 \leq \tau \leq 800$ and $100 \leq {\rm ID} \leq 220$ in (a).
(d) Heat map of the TBCC values in the same area as shown in (c).}
\label{fig:heatmap-TCC}
\end{figure*}

\subsection{Delay caused by removing a central temporal vertex}

In closing this section, to verify the relevance of the proposed centrality notions at the microscopic level, we briefly report that removing a temporal vertex with large TCC and TBCC values is effective in delaying the propagation of information.

Let $G = (V,E)$ be a temporal network, where $V = \{v_1,v_2,\ldots,v_n\}$.
For a temporal vertex $\biv = (v,\tau)$,
let $\biv_i = (v_i, \eat(\biv,v_i))$ for each $i \in [n]$ and $\tau'$ be the (unique) time such that $\biv$ has an edge to $\biv' = (v,\tau')$.
We say that $\biv_i$ gets prolonged by removing $\biv$ if $\eat(\biv, v_i)$ becomes larger by removing edges incident to $\biv$ (and we keep edge $(\biv,\biv')$).
In a similar manner, we say that $\biv_i$ becomes disconnected by removing $\biv$ if we cannot reach $\biv_i$ from $\biv$ after removing edges incident to $\biv$ (where, again, we keep edge $(\biv,\biv')$).

We investigate the fraction of prolonged or disconnected temporal vertices among $\biv_1,\biv_2,\ldots,\biv_n$, by removing one of the top $100$ vertices with respect to the TCC or TBCC values.
It should be noted that the fraction of temporal vertices becoming prolonged or disconnected is nontrivial because the definition of TCC and TBCC take into account temporal paths both before and after the focal temporal vertex. As a baseline for comparison, we also conduct the same test by removing a temporal vertex chosen randomly.
For the random case, we randomly choose $100$ temporal vertices without replacement and take the average of the fraction of prolonged or disconnected temporal vertices for these $100$ trials.

The results of the removal test of temporal vertices are summarized in Table~\ref{tbl:delay} for the five datasets.
As we expected, the removals according to the largest centrality values make more temporal vertices prolonged or disconnected than the random removals.
The removals according to the largest TCC values tend to prolong a certain fraction of temporal vertices for all the datasets considered. However, it makes few temporal vertices disconnected. These outcomes make sense because the number of other temporal paths running alongside the temporal path going through the focal temporal vertex is not considered in TCC (also see Section~\ref{sec:temporal-coverage-centrality}).
By contrast, the removals according to the largest TBCC values make a considerable fraction of temporal vertices prolonged and disconnected.
Remarkably, $50.8\%$ of the temporal vertices, on average, become disconnected from a removed temporal vertex in Irvine.
There is no clear distinction between the results of the offline (i.e., Infectious, HT09, and Hospital) and online (i.e., Irvine and Email) networks.   

\begin{table*}[t]
  \centering
  \caption{Results of the removal of temporal vertices. The number in each cell presents the average fraction of disconnected (or prolonged) temporal vertices over the $100$ trials of the removal based on the given procedure (i.e., according to the largest TCC and TBCC values or random pick).}
  \label{tbl:delay}
  \begin{tabular}{|c|cc|cc|cc|}
  \hline
  \multirow{2}{*}{Dataset} & \multicolumn{2}{c|}{TCC} & \multicolumn{2}{c|}{TBCC} & \multicolumn{2}{c|}{Random} \\
  & Prolonged & Disconnected & Prolonged & Disconnected & Prolonged & Disconnected \\
  \hline
  \hline
  Infectious & 0.013 & 0.001 & 0.014 & 0.232 & 0.010 & 0.001 \\
  HT09 & 0.082 & 0.001 & 0.264 & 0.069 & 0.031 & 0.007 \\
  Hospital & 0.049 & 0.001 & 0.156 & 0.257 & 0.037 & 0.001 \\
  Irvine & 0.014 & 0.003 & 0.006 & 0.508 & 0.018 & 0.012 \\
  Email & 0.136 & 0.006 & 0.375 & 0.016 & 0.054 & 0.000 \\
  \hline
  \end{tabular}
\end{table*}

\section{Conclusions}\label{sec:conclusion}
We introduced two centrality notions for temporal networks---temporal coverage centrality and temporal boundary coverage centrality---to represent the importance of a temporal vertex by the fraction of vertex pairs that can or should use the temporal vertex when sending information as quickly as possible.
Compared to centrality notions proposed in previous work, TCC and TBCC have two advantages: (i) Parameters or time windows do not need to be set and (ii) computation time is reasonable.

Applying TCC and TBCC to multiple datasets of empirical temporal networks, we revealed that there tends to be particular bottleneck time periods that play a crucial role in propagating information quickly and that the rest of the networks is redundant in the sense that there are many temporal paths to send information with the same duration. Although such structural redundancy in temporal networks was suggested in some previous studies~\cite{Trajanovski:2012,Takaguchi:2012,Scellano:2013}, our centrality notions enable us to clearly quantify and visualize this property.
We believe that the centrality notions we proposed are useful for further studying the structure of temporal networks and verifying generative models of temporal networks.

Datasets used in the numerical experiments, Infectious, HT09, and Hospital were originally collected and published by the SocioPatterns collaboration (\url{http://www.sociopatterns.org/}).
Datasets HT09 and Hospital were downloaded from the SocioPatterns website.  
Datasets Infectious, Irvine, and Email were downloaded from the Koblenz Network Collection (\url{http://konect.uni-koblenz.de/}).
The authors thank Dr.~James Cheng for valuable discussions. 
Yuichi Yoshida is supported by JSPS Grant-in-Aid for Young Scientists (B) (No.~26730009), MEXT Grant-in-Aid for Scientific Research on Innovative Areas (24106003), and JST, ERATO, Kawarabayashi Large Graph Project.
T.T., Y. Yano and Y. Yoshida designed the research. Y. Yoshida constructed the algorithms to compute the centralities and gave the proof of their computational complexity. Y. Yano implemented the algorithms. T.T. analyzed the data sets. Y. Yano performed the numerical experiments of the removal of temporal vertices. T.T., Y. Yano, and Y. Yoshida discussed all the results and wrote the manuscript.

\appendix
\section{Computational complexity of calculating $\eat$ and $\ldt$ with the reachability oracle}

With the aid of the reachability oracle,
we can efficiently compute $\eat$ and $\ldt$:
\begin{lemma}\label{lem:compute-eat-and-ldt}
  Let $G$ be a temporal network and $\widehat{G}$ be its DAG representation.
  We can compute $\eat$ and $\ldt$ with $O(\log |E|)$ queries to the reachability oracle of $\widehat{G}$.
\end{lemma}
\begin{proof}
  We only consider $\eat$ as $\ldt$ can be computed similarly.
  Given temporal vertex $\biv$ and vertex $w$,
  $\eat(\biv,w)$ is the minimum $\tau \in \bbR$ such that there is a temporal path from $\biv$ to $(w,\tau)$.
  To find such $\tau$, we perform a binary search using the reachability oracle.
  Since the number of possible values for $\tau$ is $O(|E|)$,
  the number of queries is $O(\log |E|)$.
\end{proof}
\begin{lemma}\label{lem:compute-TCC-and-TBCC}
  Let $G$ be a temporal network and $\widehat{G}$ be its DAG representation.
  For any temporal vertex $\biv$,
  we can compute the TCC and TBCC values of $\biv$ with $O(|V|^2 \log|E|)$ queries to the reachability oracle of $\widehat{G}$.
\end{lemma}
\begin{proof}
  The proof is immediate from Lemma~\ref{lem:compute-eat-and-ldt} and the algorithm definitions of TCC (Algorithm~\ref{alg:temporal-coverage-centrality}) and TBCC (Algorithm~\ref{alg:temporal-boundary-coverage-centrality}).
\end{proof}

\section{Approximate computation of temporal coverage centralities}
By Lemma~\ref{lem:compute-TCC-and-TBCC} (see Section \ref{sec:compute}),
the number of queries to the reachability oracle for computing the TCC and TBCC values is (almost) quadratic in the number of vertices of a temporal network.
However, in some applications, we may want to compute these centralities faster.
Here, we introduce a standard technique that enables us to approximate these centrality values with a sublinear number of queries.
We only explain the case of TCC; the case of TBCC is performed in a similar way.

Algorithm~\ref{alg:approximate-temporal-coverage-centrality} is an approximate method for computing the centrality value.
The difference from Algorithm~\ref{alg:temporal-coverage-centrality} is that,
instead of enumerating all pairs $(u,w)$,
we only sample $O(1/\epsilon^2)$ pairs of vertices and take the average over them, where $\epsilon$ is the parameter controlling the possible error in approximation.
\begin{algorithm}[!t]
  \caption{(Approximation to the TCC value of $\biv$)}
  \label{alg:approximate-temporal-coverage-centrality}
  \begin{algorithmic}[1]
    \STATE $r \leftarrow 0$.
    \FOR{$i = 1$ to $k := \frac{1}{2\epsilon^2}\log (2|V|^2)$}
      \STATE Sample vertices $u,w \in V$ uniformly.
      \STATE $\biu \leftarrow (u,\ldt(\biv,u))$.
      \STATE $\biw \leftarrow (w,\eat(\biv,w))$.
      \IF{$\eat(\biu, w) = \biw$ and $\ldt(\biw,u) = \biu$}
        \STATE $r \leftarrow r + 1$. \label{line:increment}
      \ENDIF
    \ENDFOR
    \textbf{return} $r / k$.
  \end{algorithmic}
\end{algorithm}

To show that Algorithm~\ref{alg:approximate-temporal-coverage-centrality} gives a good approximation, we need to recall Hoeffding's inequality:
\begin{lemma}[Hoeffding's inequality \cite{Hoeffding:1963}]\label{lem:hoeffding}

\noindent
  Let $X_1,X_2,\ldots,X_k$ be independent random variables in $[0,1]$ and $\overline{X} = (1/k)\sum_{i=1}^k X_i$.
  Then, for any positive real number $t$,
  \[
    \Pr[| \overline{X} - \E[\overline{X}] | \geq t] \leq 2\exp(-2t^2 k).
  \]
\end{lemma}

\begin{lemma}\label{lem:approximate-TCC-and-TBCC}
  Let $G$ be a temporal network and $\widehat{G}$ be its DAG representation.
  For any temporal vertex $\biv$,
  with $O(\log^2 |V| / \epsilon^2)$ queries to the reachability oracle of $\widehat{G}$,
  we can compute the TCC value of $\biv$ with additive error of $\epsilon$ with probability of at least $1 - 1/|V|^2$.
\end{lemma}
\begin{proof}
  Consider Algorithm~\ref{alg:approximate-temporal-coverage-centrality} and let $\widetilde{\TC}(\biv)$ denote its output.
  Algorithm~\ref{alg:approximate-temporal-coverage-centrality} issues ${\rm O}(\log^2 |V| / \epsilon^2)$ queries since $\ldt$ and $\eat$ can be computed with ${\rm O}(\log |V|)$ queries (see Lemma~\ref{lem:compute-eat-and-ldt}).
  Let $X_i$ be the temporal edge at which we increment $r$ in the $i$-th loop and $\overline{X} = (1 / k) \sum_{i=1}^k X_i$.
  Note that $\E[\widetilde{\TC}(\biv)] = \E[\overline{X}] = (1 / k)\sum_{i=1}^k \E[X_i] = \TC(\biv)$,
  where $\TC(\biv)$ is the TCC value of $\biv$.
  Since $X_1,X_2,\ldots,X_k$ are independent random variables in $[0,1]$,
  by Lemma~\ref{lem:hoeffding},
  we have
  \begin{align*}
    & \Pr[|\widetilde{\TC}(\biv) - \TC(\biv) | \geq \epsilon]
    = \Pr[|\overline{X} - \TC(\biv) | \geq \epsilon] \\
    & \leq 2\exp(-2 \epsilon^2 \frac{1}{2\epsilon^2}\log(2|V|^2))
    = 2\exp(-\log(2|V|^2)) \\
    & = \frac{1}{|V|^2}.
  \end{align*}
  Hence, the lemma holds.
\end{proof}

Recalling that the query time of the reachability oracle is tiny,
we find that the running time of Algorithms~\ref{alg:approximate-temporal-coverage-centrality} can be seen as polylogarithmic in the input size.
This is the great advantage of TCC and TBCC against other centrality notions.

\bibliographystyle{unsrt}

\clearpage
\twocolumn[
\setcounter{page}{1}
\setcounter{figure}{0}
\begin{center}
{\bf {\Large Electronic Supplementary Material}}\\
\vspace*{3mm}
{\large for}\\
\vspace*{3mm}
\textit{{\large Taro Takaguchi, Yosuke Yano, and Yuichi Yoshida}}\\
\vspace*{3mm}
{\bf {\large Coverage centralities for temporal networks}}\\
\end{center}
]

\newpage

\renewcommand{\thesection}{S\arabic{section}}
\renewcommand{\thefigure}{S\arabic{figure}}
\renewcommand{\thetable}{S\arabic{table}}
\renewcommand{\theequation}{S.\arabic{equation}}

\begin{figure*}
\centering
\includegraphics[width=0.48\hsize]{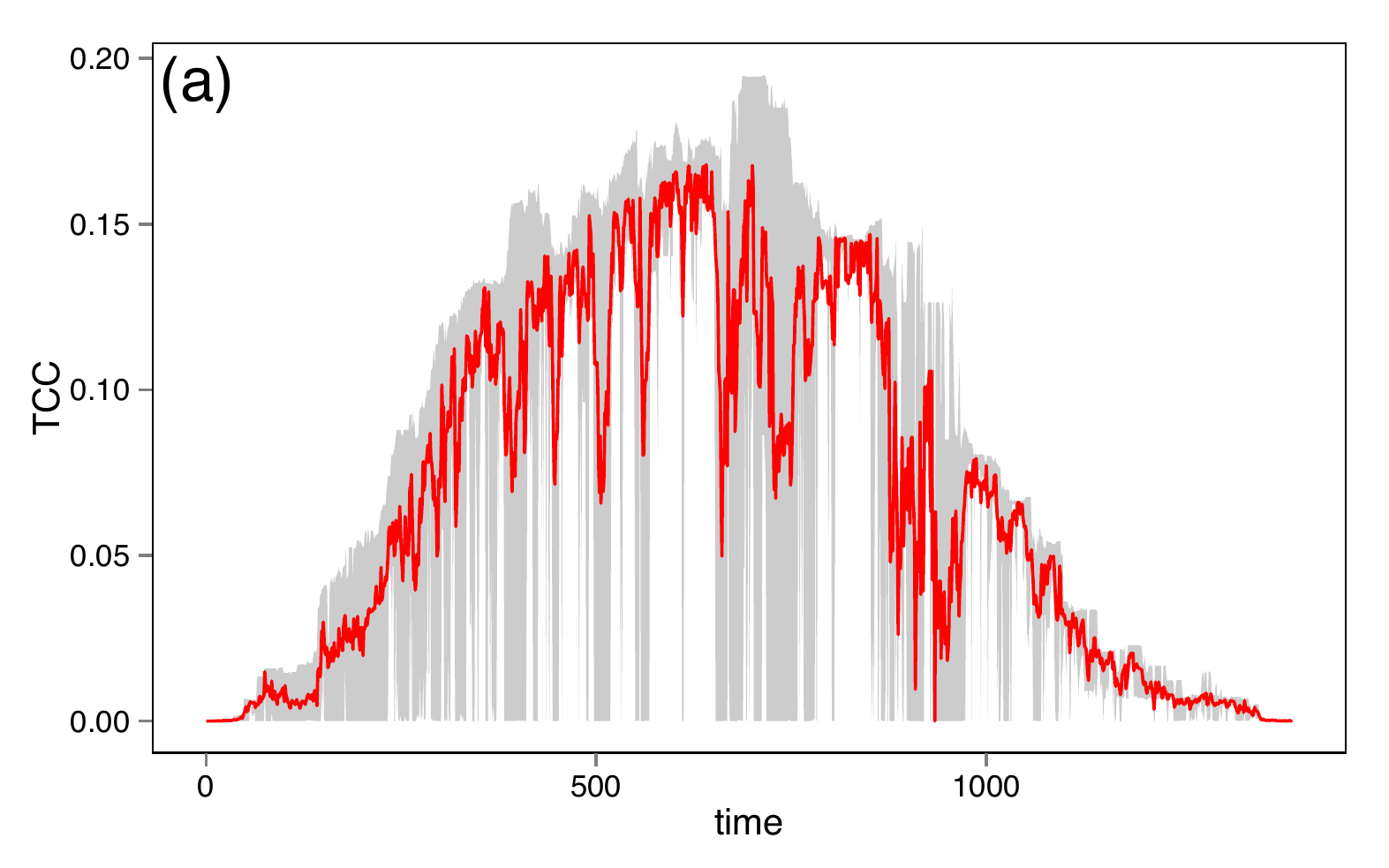}
\includegraphics[width=0.48\hsize]{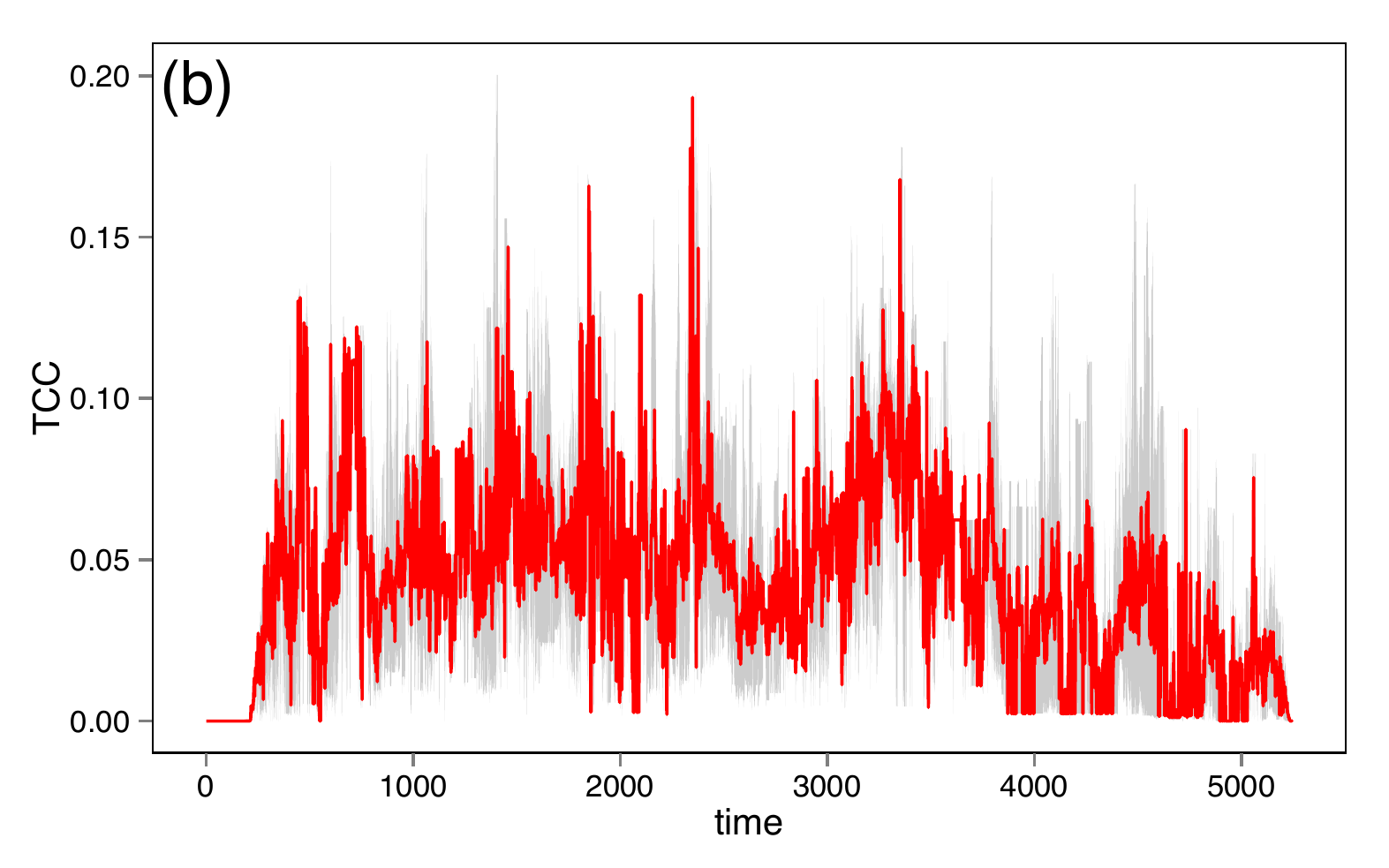}
\includegraphics[width=0.48\hsize]{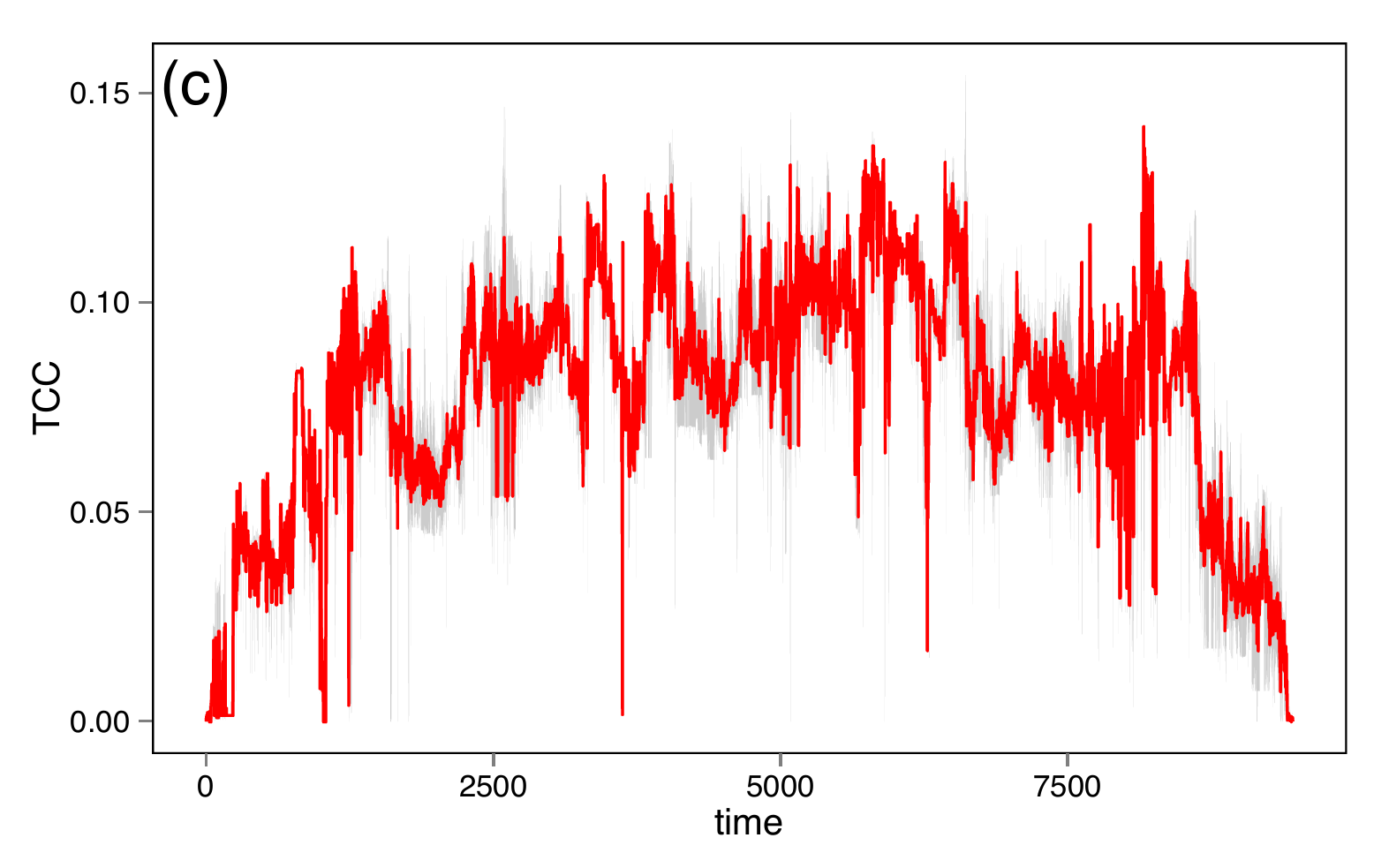}
\includegraphics[width=0.48\hsize]{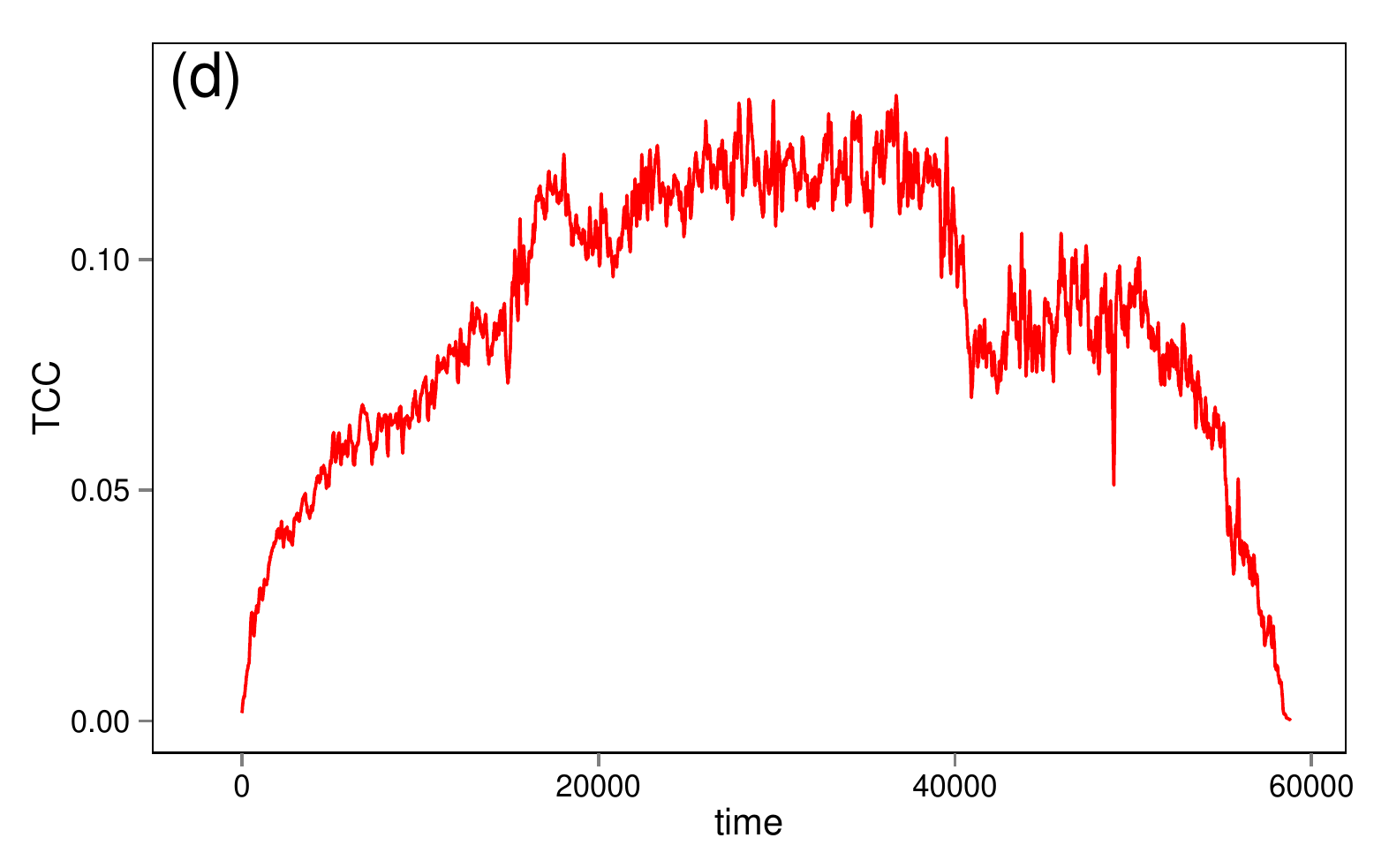}
\includegraphics[width=0.48\hsize]{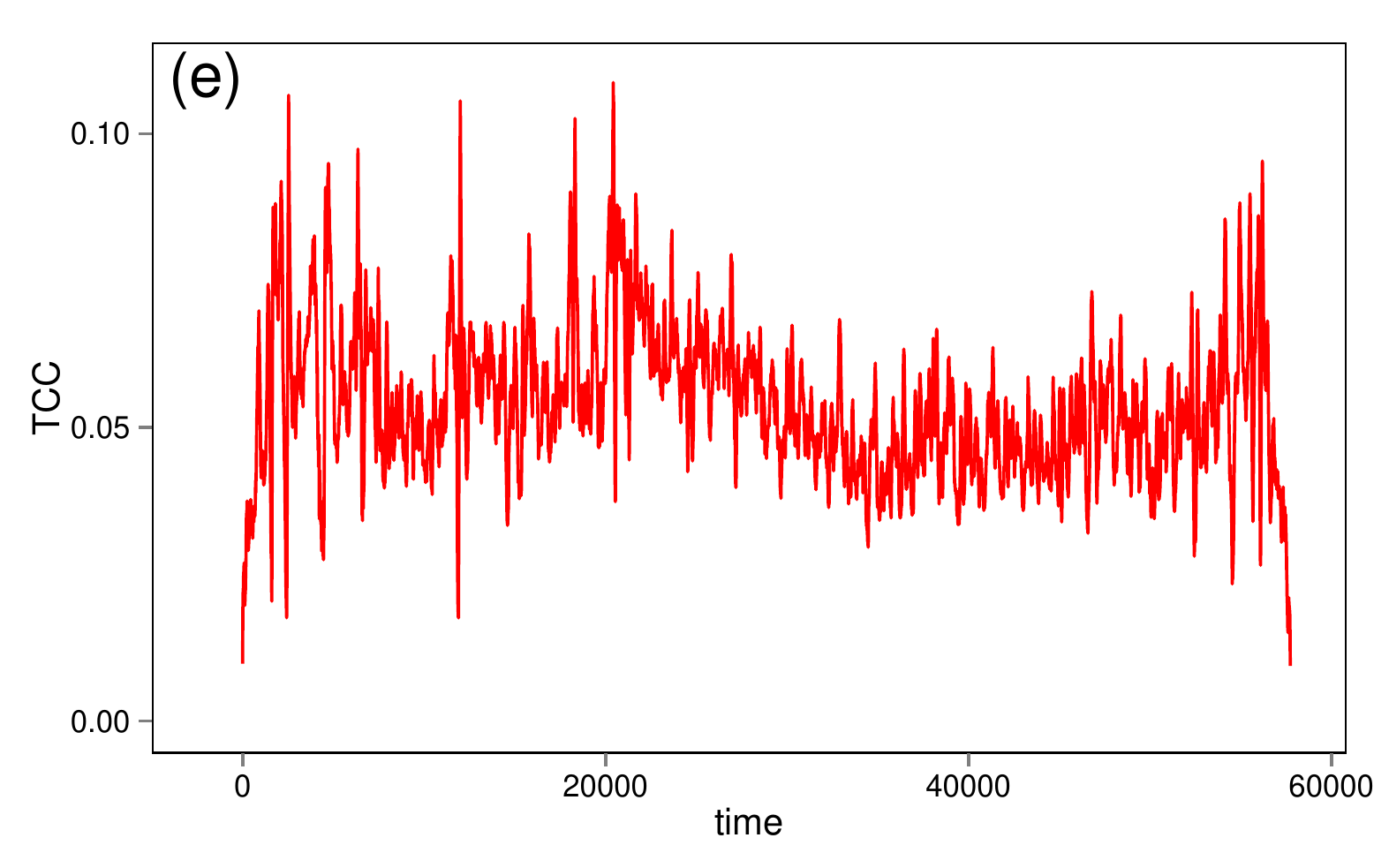}
\caption{Average (solid line) and $10-90\%$ values (shaded areas) of TCC at each time for (a) Infectious, (b) HT09, (c) Hospital, (d) Irvine, and (e) Email.
We consider only the temporal vertices involved in temporal edges with other vertices to calculate the statistics.
For (d) and (e), we smoothed the curves by taking the average over a sliding window with a length of 100 units of time, because the time resolutions of the observations are so high that there are not sufficient number of temporal vertices to take the average at most of the time points.}
\end{figure*}

\end{document}